\pdfoutput=1 
\documentclass[a4paper,USenglish]{lipics-v2019}

\title{Comparison Graphs: a Unified Method for Uniformity Testing}

\author
{Uri Meir}
{Tel Aviv University, Israel}
{}
{}
{}

\authorrunning{U. Meir} 

\Copyright{Uri Meir} 



\relatedversion{} 

\supplement{}

\acknowledgements{The author would like to thank Rotem Oshman for valuable counsel, Noam Mazor for very helpful discussions, and Ami Paz as well as the anonymous referees of ITCS 2021 for their useful comments.}

\nolinenumbers 

\hideLIPIcs  

\EventEditors{James R. Lee}
\EventNoEds{1}
\EventLongTitle{12th Innovations in Theoretical Computer Science Conference (ITCS 2021)}
\EventShortTitle{ITCS 2021}
\EventAcronym{ITCS}
\EventYear{2021}
\EventDate{January 6--8, 2021}
\EventLocation{Virtual Conference}
\EventLogo{}
\SeriesVolume{185}
\ArticleNo{19}

\usepackage{amsthm}
\usepackage{amsmath}
\usepackage{amssymb}
\usepackage{color}
\usepackage{bm}
\usepackage{dsfont}

\renewcommand{\phi}{\varphi}
\newcommand{\set}[1]{\left\{#1\right\}}

\renewcommand{\st}{\medspace | \medspace}

\newcommand{\eps}{\epsilon}
\newcommand{\norm}[1]{\left\| #1 \right\|}
\newcommand{\lnorm}[2]{\left\| #2 \right\|_{#1}}

\DeclareMathOperator{\poly}{poly}

\DeclareMathOperator*{\Var}{Var}

\newcommand{\ind}[1]{\mathds{1}_{#1}}

\newcommand{\EE}{\mathbb{E}}
\newcommand{\card}[1]{\left| #1 \right|}

\renewcommand{\mathbf}{\bm}

\newcommand{\alg}[1]{\ensuremath{\mathcal{#1}}}

 \newtheorem{conjecture}[theorem]{Conjecture}
\newtheorem{observation}{Observation}
\newenvironment{lemma-repeat}[1]{\begin{trivlist}
		\item[\hspace{\labelsep}{\bf\noindent Lemma \ref{#1} }]\em }%
	{\end{trivlist}}
\newenvironment{theorem-repeat}[1]{\begin{trivlist}
		\item[\hspace{\labelsep}{\bf\noindent Theorem \ref{#1} }]\em }%
	{\end{trivlist}}

\begin{document}
\maketitle

\begin{abstract}
    Distribution testing can be described as follows: $q$ samples are being drawn from some unknown distribution $P$ over a known
    domain $[n]$. After the sampling process, a decision must be made about whether $P$ holds some property, or is far from it.
    The most studied problem in the field is arguably \emph{uniformity testing}, where one needs to distinguish the case that $P$ is uniform over $[n]$ from the case that $P$ is $\eps$-far from being uniform (in $\ell_1$). The \emph{sample complexity} of a property is the amount of necessary and sufficient samples to test for this property, and it is known to be $\Theta\left(\sqrt{n}/\eps^2\right)$ when testing for uniformity. This problem was recently considered in various restricted models that pose, for example, communication or memory constraints. In more than one occasion, the known optimal solution boils down to counting collisions among the drawn samples (each two samples that have the same value add one to the count). This idea dates back to the first uniformity tester, and was coined the name ``collision-based tester''.

    In this paper, we introduce the notion of \emph{comparison graphs} and use it to formally define a generalized collision-based tester. Roughly speaking, the edges of the graph indicate the tester which pairs of samples should be compared (that is, the original tester is induced by a clique, where all pairs are being compared).
    We prove a structural theorem that gives a sufficient condition for a comparison graph to induce a good uniformity tester.
    As an application, we develop a generic method to test uniformity, and devise nearly-optimal uniformity testers under various computational constraints. We improve and simplify a few known results, and introduce a new model in which the method also produces an efficient tester.

    The idea behind our method is to translate the computational constraints of a certain model to ones on the comparison graph, which paves the way to finding a good graph: a set of pairs that can be compared in this model, and induces a uniformity tester.
    We believe that in future consideration of uniformity testing in new models, our method can be used to obtain efficient testers with minimal effort.
\end{abstract}

\section{Introduction}
\label{sec:intro}
The field of property testing was initiated by~\cite{BLR93,RS96,GGR98} and concerns with fast probabilistic algorithms that use query access to some large structure (such as graphs, functions, distribution, etc.) in order to determines whether a specific instance belong to a subclass of possible instances (e.g., connected graphs or monotone functions) or in some sense far from it.
Specifically for distributions, as formulated in~\cite{BFRSW00}, we are given random samples from some unknown distribution, and we wish to decide with high probability (over the random samples) whether it has some property, or it is far from any distribution that does (typically we use $\ell_1$ as a distance measure).
Properties of distributions were excessively studied through the years (see~\cite{Goldreich17,Can15} for excellent surveys). Until recently, the vast majority of results were limited to the classic setting, where a \emph{single}
processor is given an oracle access, and performs the testing procedure. The measure of complexity is based solely on the number of samples, as typically the running time is polynomial in this number.

However, distribution testing can be very useful in different frameworks as well.
For example, suppose we have a sensor network taking some measurements that need to be combined in order to make a decision about the subject of these measurements, be it volcanic activity, seismic movements, or any form of radiation.

Another framework where testing is useful, is under constrained memory. One might try testing an object so big, that the number of samples needed is very large. In that scenario, even if one might endure a lengthly sampling process, storing all past samples at one given moment can be too costly. For example, imagine a large telescope collecting data on infrared radiation in a pursuit to discover new planets. These types of questions can be translated to a \emph{streaming model}, where samples come as a stream, and one wish to store only a small amount of data while reliably test for a property of the underlying distribution.

When considering the relatively similar motivations for these cases, one might even wonder about a combination of the two models, where multiple sensors spread out in some area collect samples, each of which has bounded memory. They will then need to process all available data within their memory constraints, such that by the end of each time period - they are able to report their individual findings concisely to some data center that aggregates all information in order to make an important decision.

All of these questions are inherently multi-dimensional, in the sense that many incomparable resources are in play: the number of players, the number of samples, the memory space of each sensor, and even the amount of bits communicated.
Here, we focus our efforts, for the most part, on minimizing the \emph{sample complexity} (or amount of samples per player, when multiple players are involved), with other resources given as parameters.

This specification is well-motivated by the case where we have no shortage of data we can sample from, and we wish to understand how long a sampling process should take, as a parameter of the number of sensors we use and their computational strength.
Throughout the text, we focus on the task of uniformity testing, a key problem in the field of distribution testing.

\subparagraph*{Uniformity testing.} The most studied family of problems in distribution testing is arguably \emph{identity testing}, where we want to test whether the input distribution $P$ is equal to some fixed distribution $p$, or $\eps$-far
from it (in $\ell_1$), where $\eps$ is the proximity parameter of the problem.
In the heart of these problems stands the problem of \emph{uniformity testing}, where $p=U_{\Omega}$.
One evidence for the importance of uniformity testing was shown in~\cite{DK16} and made more robust in~\cite{G16}: it is actually complete for identity testing, in the sense that testing identity to any fixed distribution $p$ can be reduced to testing uniformity instead.
On the other hand, uniformity testing is a specific case of other problems such as closeness testing and independence testing,
so showing lower bounds for uniformity testing would imply lower bounds for these problems.
The classic version of testing uniformity was settled for the case $\eps = \Omega\left(n^{-1/4}\right)$ in~\cite{Paninski08}, showing that $\Theta\left( \sqrt{n}/\eps^2 \right)$ samples are in fact sufficient and necessary. \cite{VV17} later showed this holds for all values of $\eps$.

\subparagraph*{Collision-based testers.} The problem of uniformity testing was implicitly introduced in the paper~\cite{GR00}, as a way to test the expansion of a graph: when simulating multiple short random walks, and observing the distribution of the endpoint, one can connect a uniform distribution over these endpoint to good expansion of the graph.
To solve uniformity, the \emph{collision-based} tester was introduced, where one simply counts the number of pairs of samples that have the same value. This tester was shown to have sample complexity of $\Theta\left(\sqrt{n} \cdot \poly(1/\eps) \right)$, which is turned out to be sub-optimal in terms of $\eps$.
However, several years after the question was settled, it was shown in~\cite{DGPP16} that the original collision-based tester also achieves optimal sample complexity, using a finer analysis.

The idea behind the collision-based tester is rather straightforward: when comparing two samples from a distribution, the chance of both having the same value (also referred as \emph{collision probability}) relates to the $\ell_2$ norm of the distribution. It is a well-known fact that over a fixed set $\Omega$, the uniform distribution has the minimal $\ell_2$ norm. It is also rather easy to show that any distribution that is somewhat far from uniform (in statistical distance), has a significantly larger $\ell_2$ norm. This means that each comparison of two samples is an unbiased estimator of the collision probability (having the right expectation), but with very high variance. One would need to average over many such comparisons in order to reduce the variance.

\subparagraph*{Other methods to test for uniformity}
Over the years, numerous methods to test for uniformity have been proposed. Some of them had different goals in mind, such as testing with very high confidence, or in a multiparty model where each player gets a single sample (sometimes even wishing to keep it private). These methods include counting unique element~\cite{Paninski08}, modified $\chi^2$ test~\cite{VV17}, using the empirical distance to uniformity~\cite{DGPP17}, randomly hashing samples to a smaller domain~\cite{ACT20B} and more.
In both~\cite{DGKR19,FMO18}, testers that aim to overcome different constraints relied strongly on collision counting\footnote{These testers do not count collisions per se, but they add other steps to the algorithm. The results also use a more involved analysis}. Considering it is also optimal in the classic setting, this makes collision-based testing a prime candidate for a more generic method to test uniformity, and hopefully adjust itself to different models easily.

\subparagraph*{Comparison graphs.}
The original version of the collision-based tester takes a set of samples, and use comparisons between \emph{all pairs of samples}. This is well-suited for the classic model, where one processor with no concrete limitations sees all the samples, and is able to perform all possible comparisons.

However, in more constrained models, this simple task is inherently impossible. For example, a memory-constrained tester cannot store all previous samples in order to compare them with new ones.
In the simultaneous model, where each processor holds its own set of samples, a lot of communication might be needed to compare samples that are held by different processors.\footnote{For a specific pair of samples, this would require solving equality. When observing the two \emph{sets} of samples held by the two processors, one can imagine a disjointness-type problem, which intuitively should be hard to solve accurately.}
In distributed models, such as CONGEST and LOCAL (see~\cite{FMO18}), the problem is defined where each player in a network holds a single sample from a distribution (replacing one sample by a constant amount produces similar behaviour). In these models, it is much cheaper for player to compare their samples with those of a neighboring player, than it is to make such a comparison with players that are far away (on the network topology).

To this end, we introduce the notion of \emph{comparison graphs}. A comparison graph is linked to a collision-based tester (or algorithm) as follows: the vertices of the graph are the samples given as input, and the edges are pairs of samples that are being compared.
As stated above: in the classic model this graph is typically the complete clique (all pairs of samples are compared). Under constrained models, however, very specific edges (comparisons) are allowed, whereas others are not. For example, if the sample $s_1$ is given to one player, and the sample $s_2$ is given to another player in the simultaneous model, no algorithm can presume to compare the two samples.

Equipped with the notion of comparison graphs, one can define a \emph{collision-based} tester as a couple $(G,\tau)$, where $G=(V,E)$ is the comparison graph that defines which comparisons are being made, and $\tau$ is a threshold parameter. The algorithm is defined as follows: it counts the amount of collisions observed, $Z$, and compares it to a threshold $T$ that depends on $\tau$ and the amount of comparisons made (which is $\card{E}$).

\subparagraph*{Reliable collisions-based testers}
In Section~\ref{subsec:thm_pf}, we prove a structural theorem concerning with which sets of comparisons are able to inspire a reliable test for uniformity based solely on counting collisions.
This is done by observing the comparison graph $G$. We formulate sufficient conditions in terms of the graph $G$, that guarantee it induces a good tester (when paired with the right threshold parameter $\tau$).
It turns out that two properties of a comparison graph $G$ are key: the first is the number of edges, which represents the amount of comparisons being made; the second one, somewhat surprisingly, is the number of $2$-paths in the graph $G$, which encapsulates the amount of dependencies between different comparisons being made.

Few of our testers rely on the same type of graph, that pops up multiple times, for different reasons. To this end, we formulate Lemma~\ref{lem:disj_cliques}, that specifies the required parameters for a comparison graph of this type to induce a good tester.

\subsection{Examples of comparison graphs}
\label{subsec:examples}
It is interesting that the number of samples (the measure we usually wish to minimize) does \emph{not} appear as a condition on our comparison graph directly. However, as we see later in the text, it does play a role indirectly, as simple inequalities connect the three graph quantities (see Section~\ref{subsec:graph_prop}).

Our structural theorem basically shows that any comparison graph with enough edges, but not-too-many $2$-paths induces a good uniformity tester. To better understand the meaning of this, we fix the amount of edges, $\card{E}$, and review a short list of examples for potential comparison graphs. We are interested in the interplay between the amount of vertices (samples), edges (comparisons made) and $2$-paths (dependencies created).

\subparagraph*{The clique graph.}
The standard tester is actually the full clique, comparing each possible pair of samples. In this dense graph we only need $\sqrt{\card{E}}$ vertices in order to have $\card{E}$ edges. However, many $2$-paths (and dependencies) are created along the way as well, $\Theta \left( \card{E}^{3/2} \right)$. For this specific case, it is already known the two affects can be balanced to obtain optimal (asymptotic) sample complexity.

\subparagraph*{Disjoint cliques.}
Another interesting graph (used in some sense in \cite{FMO18}) is actually a union of disjoint cliques. This graph turns out to be quite useful. For once, it makes perfect sense in a simultaneous model, where each player process her own samples, sending a short summary to the referee. Surprisingly, it arises in other models as well.

\subparagraph*{A perfect matching.}
This graph relates to taking a fresh pair of samples each time we wish to make a new comparison, which leaves us with a set of completely independent collision indicators. Indeed, in this graph there are no $2$-paths at all.
Not only this tester minimizes the dependencies -- it actually overdoes it. Doing so, it pays a price in sample complexity: the number of vertices we have is $2\card{E}$, much larger than the clique, for instance.

\subparagraph*{The star graph.}
With a fixed number of edges, this graph is actually the way to \emph{maximize} the amount of $2$-paths and dependencies -- which makes it a very poor comparison graph.
This makes perfect sense, as the tester described is equivalent to drawing one element from $P$, and comparing it to many other samples, basically assessing the probability of this element. This test can be shown to perform poorly against specific distributions. e.g., if $P$ has probability $1/n$ for \emph{most} elements. In this case we are very likely to draw such an element first, and from here on every comparison has a $1/n$ chance to show a collision. This would make $P$ indistinguishable from the uniform distribution, even though it might be very far from it in practice.

\subparagraph*{The full bipartite graph.}
Another graph that could be considered is the full bipartite graph, $G = (V_1 \sqcup V_2, V_1\times V_2)$.
Here again we have a free parameter (the size $\card{V_1}$, which determines $\card{V_2} = \card{E} / \card{V_1}$).
It ranges from a star-graph (for $\card{V_1} = 1$) to a balanced graph (for $\card{V_1} = \card{V_2} = \sqrt{E}$), where the last one functions asymptotically similarly to the clique. It appears that optimal results for these graphs can be obtained through our framework only for the balanced case.\footnote{By optimal here, we mean having ``just enough'' comparisons, but no more. Once fixing $\card{E}$, a more imbalanced graph admits more dependencies. Whenever $V_1 = \omega(\sqrt{\card{E}})$, too many dependencies are created and the theorem cannot be applied. This statement can be formalized in a similar fashion to the results in Section~\ref{sec:lower}. }
To some extent, such a tester was used in~\cite{DGKR19} to test uniformity in the streaming model, storing the first batch of samples ($V_1$) then comparing the rest of the stream ($V_2$) to this batch. Their result extends to a large range of imbalanced graphs even when the number of comparison is small. However, their tester in fact does a little more than counting collisions, as it first examines the entire sample set $V_1$ and in some cases decides to abort. This additional step is also integrated strongly in the analysis, which leads us to believe the tester would not achieve its goal without it. In Section~\ref{subsec:res_streaming_processor}, it is shown that one can test uniformity in the streaming model via collisions with no additional steps, using a whole different comparison graph.

\medskip
The examples raise another inherent question: is it always better to add edges in a graph with a fixed number of vertices (that is, after making a set of comparisons $E$, can it hurt to add more comparisons to the calculation?)
Intuitively, adding edges seems to only improve the performance of the algorithm. Proving a statement of the sort could better our understanding and point towards optimal comparison graphs under certain constraints.
In Appendix~\ref{sec:counter_example} we discuss the matter and explain why such a statement is not true, at least if one goes through our structural theorem to prove correctness of a collision-based tester.

\subsection{Models and Results}
Our main result is a method that produces well-performing uniformity testers in various models. In this paper we show a list of uniformity testers in different models, specified below. We also show limitations of our method for most of these models, which point towards a conclusion that no better comparison graphs could have been chosen (up to constant factors in the sample complexity of the induced tester). These limitation rely on a conjecture that in some sense no matter the shape of the comparison graph, enough comparisons always must be made.

We emphasize that for any model in which a lower bound is known (for any method, not necessarily collision-based testing), the testers produced by our method are optimal, up to $\poly(1/\eps)$ factors. As far as we know, no testers in the literature are tight with the current lower bounds for these specific cases. Thus, it could be the case that collision-based testing achieves optimal results (even in terms of $\eps$) for all the models we consider. A more thorough discussion is given in Section~\ref{subsec:discussion}.

Equipped with a structural theorem, proven in Section~\ref{subsec:thm_pf}, we consider various models and devise a uniformity tester in each one. The key idea here is to translate the constraints of each model into the comparisons we are able to perform, or differently put: a structural limitation on the comparison graph.
Doing so will guide us how to choose a ``good'' comparison graph for this specific model.
Having a structure in mind, two formalities are left: (i) Prove that calculation of $Z, T$ (the number of collisions, and the threshold value) can be done in the model; (ii) Calculate our desired complexity measure (which changes from model to model), and optimize the parameters of the chosen graph (e.g., if the graph is a clique, determine the size of the clique).

\subparagraph*{Standard processors.}
In Section~\ref{subsec:res_standard_processors}, we deal with the classic and the simultaneous model.
First we use the classic model as a warm-up. Since this is done in ~\cite{DGPP16}, we add to the mix a small insight: our framework (which allows the tester to choose threshold values other than $\tau = 1/2$) actually provides slightly better constants when placing the threshold much lower (at roughly $\tau = 1/9$).

For this model, it was shown by~\cite{DGPP17} that testing with high precision can be done faster than it would have using standard amplification. However, they use estimation of the empirical distance from uniform over the sample set. It is unclear whether collision-based testing is fit for this task. We do not pursue this direction here.\footnote{One reason is that our work focuses on the regime that uses new and different comparison graphs, other than the clique graph. This direction would probably involve analysis that is specific for the clique graph, where all samples can be compared with one another.}

We then move to the simultaneous case where multiple players each send a short message to the referee, based on their own samples. The referee then needs to output a decision about the underlying distribution. In this model we want to find a good exchange for the number of players, the number of samples each player gets, and the length of the messages. For example, if the messages can be arbitrarily long, each player can send her entire sample set and the problem becomes trivial.
For this reason, the two papers to first consider (independently) testing in this model, had a very different focus. In a preliminary version of~\cite{ACT20A,ACT20B}, only the case of a \emph{single sample} per player was considered, and their algorithms indeed rely on different and interesting strategies, but not collision counting -- as this strategy is irrelevant for this regime.
In~\cite{FMO18} a different approach was taken, where all messages were fixed to a \emph{single bit}, but each player gets multiple samples (and in fact, their algorithm does rely on collisions in some sense, but it does not count them accurately, and does not fall under the umbrella of our definition for collision-based testers).

For our use, as oppose to both these view, we allow both parameters to be larger. We allow multiple samples per player, and show that using a short message (not a single bit, but not much longer), one can devise an efficient tester.

\medskip

We also consider the asymmetric cost variant of this model, which naturally arise when reducing to it from the LOCAL model, a standard model in the field of distributed computing. To motivate this variant, we think of the sampling process as a bottleneck, but now each player has her own sampling rate (some players might draw samples much faster than others).
If before we tried to minimize the amount of samples per player, now we wish to minimize the overall sampling time instead, and still sample enough data to successfully test for uniformity. More about the model and the reduction from the LOCAL model can be found in~\cite{FMO18}.

We remark that many works in the simultaneous model consider communication trade-offs, when assigning only a \emph{single} sample for each party. Our method does not currently extend to this framework, although one might consider integrating it with other methods. For example, one method (e.g., in~\cite{ACT20B}) uses random hash to a smaller domain. One might consider collisions on this domain instead of the original one. these meta-collisions can be counted within the communication constraints.

Some works in this regime (single sample) focus on privacy aspects of testing (e.g.,~\cite{ACFT19,AJM20}). This line of research should be irrelevant for our method (and even the extension mentioned above), as any detection of a collision (even on a smaller domain) would immediately give away non-trivial information about the sample.

\subparagraph*{Memory-constrained processors.}
In Section~\ref{subsec:res_streaming_processor}, we deal with a different type of processors: \emph{memory-constrained} processors. When observing a single processor of this type, we end up with the streaming model (as described in~\cite{DGKR19}).
At least one scenario which motivates the simultaneous model is seemingly very coherent with such constrained processors. Thinking of a network of sensors, or remote devices, gathering samples -- it is quite comprehensible that these processors are not only limited by their ability to communicate with the data center, but also by their ability to store the entire data observed between consecutive reports (in this scenario, the data center is the referee performing the test periodically to detect anomalies).
For this reason, we also consider the case of a simultaneous model, where each processor has a small memory budget. In this new model, we easily devise again an efficient uniformity tester.

\subparagraph*{Testing in an interactive model.}
Lastly, in Section~\ref{subsec:res_congest}, we use the structural theorem to show how on certain graphs one can solve uniformity in the CONGEST model faster than what was previously known to be possible.
Specifically, if the communication network has $k$ players and diameter $D$, and each players start with one sample, the best known algorithm runs in $O(D+ n/(k\eps^4))$ rounds~\cite{FMO18}.
The improvement we suggest is an algorithm that takes $O(D)$ rounds, and works in specific networks that have a certain topological characteristics.
Moreover, we show a simple detection procedure of $O(D)$ rounds, that can be used to recognize such a good topology. This means that \emph{any} network can use the detection procedure first, and then either proceed as in~\cite{FMO18}, or switch to the faster ($O(D)$) algorithm whenever it is guaranteed to perform well.

\subsection{Related Work}
The task of uniformity testing, as well as the collision-based tester for it, were introduced in~\cite{BFRSW00} (and implicitly in~\cite{GR00}). Later on, upper and lower bounds on the sample complexity of the problem were given by \cite{Paninski08} (for most values of $\eps$), showing that the optimal sample complexity is in fact $s = \Theta\left(\sqrt{n}/\eps^2\right)$, where $n$ is the world size, and $\eps$ is a proximity parameter. A preliminary version of~\cite{VV17} giving a tester that achieves this complexity for any value of $\eps$.
The last two papers used two different testers: the first relied on the number of \emph{distinct} elements in the sample set (this test is somewhat dual to counting collisions), and the second on a modified $\chi^2$ tester. It was then shown by \cite{DGPP16} that the collision-based tester does in fact achieve optimal sample complexity too.

In the past few years there has been a growing interest in distribution testing under various computational models, including collaborative testing in multiparty model, the streaming model, privacy aspects of testing, and others (\cite{ACT20A,ACT20B,FMO18,AMN18,ACFT19,MMO19,ACLST20,ACT20,ABIS19,NMP20,GKR20,AJM20} and more). We mention in more details the works concerning uniformity testers in models we pursue.

\medskip

In a preliminary version of~\cite{ACT20A,ACT20B} and~\cite{FMO18} each, independently, the task of uniformity testing was considered in a simultaneous communication model. In this model, all players receive samples from the same global distribution $P$, and all players report to a referee based on their own samples. The referee in turn uses the reports to output (with high probability) whether $P$ holds some property.

In both lines of work, the focus was uniformity testing, but using two different perspectives: the former zeroes in on one sample per player, where the trade-off in question is between the number of bits each player is allowed in his report, and the number of players needed to the testing process. We think of the number of bits as too small to describe the sampled element fully.
In this setting, a full description of two samples is never available to a single player (nor the referee, whenever no player can describe fully the element it samples). We do note that the tester given there relies on a looser notion of collisions (Taking a coarser division of $[n]$ into subsets).

In the later, the focus is different: one now fixes instead the communication to \emph{one} bit per player. Now, the trade-off is between the number of players and the number of samples each one of them takes. The upper bound devised there discuss each player taking the right amount of samples, and notifying the referee $0$ if no collision occurred, and $1$ otherwise. In some sense, the referee ends up counting collisions (notice the count is trimmed, as one player might see more than $1$ collision, but is only able to report 0 or 1). This tester is then used as a black-box to solve uniformity testing in the classic distributed models (CONGEST, and LOCAL), in a setting where each player initially draws one sample.

One other paper to specifically discuss uniformity testing is~\cite{DGKR19}, in which a streaming version of the problem is defined, as well as another distributed version, in a blackboard model, where all players are privy to the messages sent by others. As opposed to the simultaneous models mentioned above -- here several rounds of communication are allowed. As it turns out, the streaming algorithm, as well one of two algorithms suggested for the distributed version, boil down yet again to counting collisions under the limitations of the model.

In addition to these, quite a few works had the focus of showing impossibility results both in the classic and the entire variety of models. For out interest, we mention~\cite{DK16,Paninski08} for the classic version, as well as~\cite{MMO19} for the simultaneous model (with multiple samples per machine), and ~\cite{DGKR19,ACLST20} for the streaming model.

\section{Preliminaries}
Throughout this text, we discuss uniformity testing using collision-based testers.

We let $[n]:=\set{1,2,\dots,n}$, and use $\Delta([n])$ to denote the set of distributions over the set $[n]$. As we only care about the support size (rather than the values), it is enough to consider $P\in \Delta([n])$ as possible input distributions. For a distribution $P$ over $[n]$, we write $P_i$ for the probability of the $i^{th}$ element.

An $(n,\eps)$-uniformity tester is an algorithm $\alg{A}$ that given oracle access to some unknown distribution $P\in\Delta([n])$, takes $q = q(n,\eps)$ samples from $P$ and satisfy the following:
\begin{itemize}
    \item If the input distribution is $P = U_n$, the uniform distribution over $[n]$, then $\alg{A}$ outputs YES with probability at least $3/4$.

    \item If $\norm{P - U_n} \geq \eps$, which means the distribution $P$ is $\eps$-far from uniform), then $\alg{A}$ outputs NO with probability at least $3/4$.
\end{itemize}
The distance used here is $L_1$. Meaning, for two distribution $P , Q$, the distance is $\norm{P-Q} = \sum_{i=1}^{n} \card{P_i - Q_i}$.

To formalize our notion of a \emph{collision-based} tester, we take a fresh point of view of the sampling process.
The key object in our analysis is the \emph{comparison graph}, which is simply an undirected graph $G = (V,E)$, where we think about $V$ as a set of placeholders for samples and $E$ as the pairs of samples which are chosen to be compared with one another.

\begin{definition}[Sampling process, collision indicators]
Given a comparison graph $G = (V,E)$ and an input distribution $P\in \Delta([n])$, the sampling procedure $S_P$ is described as a random labeling of the vertices according to $P$.
We denote by $S_P:V \to [n]$, the process for which $\forall i. S_P(v_i) \sim P$, independently from one another. We end up with $S_P(v_1),\dots,S_P(v_{_{\card{V}}})$ which is a set of $\card{V}$ i.i.d samples from $P$.

Moreover, for every edge $e = (u,v)$ in $E$ we define a unique indicator, we call the collision indicator and denoted by $\ind{e}^{P} := \ind{S_P(u) = S_P(v)}$.

When $P$ is clear from context, we simply write $S(u)$ for the sample associated with vertex $u$, and $\ind{e}$ for the collision indicator of the edge $e$.
\end{definition}

We are now ready to give a formal definition for a collision-based tester, which relies on a set of comparisons (not necessarily between all pairs of samples), and compares $Z$ - the amount of collisions, with some threshold value $T$.

\begin{definition}[Collision-based tester]
\label{def:cb_tester}
  Fix $n, \eps$. For any comparison graph $G = (V,E)$ and real number $0\leq \tau \leq 1$, we define the algorithm $\alg{A} = (G,\tau)$ as follows: upon receiving as input $\card{V}$ i.i.d samples from $P$ (given by $S_P(v_1),\dots,S_P(v_{_{\card{V}}})$), it computes the following:
  \begin{align*}
    Z &:= \sum_{e\in E} \ind{e}\ , &
    T &:= \card{E}\cdot \left(\frac{1+\tau\eps^2}{n}\right)\ ,
  \end{align*}
  and outputs YES if $Z < T$, and NO otherwise.
\end{definition}

The restriction $\tau\in[0,1]$ will help us deal with technicalities, but we note that it is rather intuitive. Indeed, as we will see later, only for these values the expectation of $Z$ is lower than $T$ for the good input (uniform distribution), and higher than $T$ for \emph{all} bad inputs.

Throughout, we will focus on properties of the graph and of our input distribution. We denote by $\card{E_G}$, $\card{V_G}$ the number of edges and vertices in $G$, and by $c(G)$ the number of times a 2-path appears as a subgraph in $G$. We count each 2-path twice, for its 2 automorphisms, and so we need to count ``directed'' $2$-paths (so $e_1,e_2$ and $e_2,e_1$ are both counted). Formally, we can write $c(G) =\card{\set{(u,v,w)\in \binom{V}{3} \st (u,v),(v,w) \in E}}$.

For the distribution $P$  over $[n]$, with $P_i$ for the probability of the $i^{th}$ element, we denote the collision probability $\mu_P = \sum_{i=1}^{n} P_i^2$, and the three-way collision probability $\gamma_P = \sum_{i=1}^{n} P_i^3$. For brevity, whenever $G$ or $P$ are clear from context, we simply write $\card{V}$, $\card{E}$, $\mu$, $\gamma$.

\subsection{Models of Computation}
In all our results we are concerned with distribution testing (and specifically uniformity testing), and we deal with various models.
Therefore, in the following lines we specify in which way samples are taken in each model, and in what way the answer of the algorithm needs to be declared (where a good tester is the one that outputs YES (resp. NO) with high probability whenever the samples are taken from a YES (resp. NO) distribution).

\subparagraph*{The centralized model.}
The centralized model is the classic model. In this model one processor receives all samples $s_1,\dots,s_q$ (in comparison graph notations, we think of $s_i = S(v_i)$, and $q = \card{V}$), and is tasked with outputting a proper answer according to the underlying input distribution.
The complexity measure we wish to minimize in this model is $q$, the number of samples.

\subparagraph*{The simultaneous model.}
The second model we consider is the simultaneous model, where $k$ players (processors) each draw individual samples unseen by all other players. Each player sends a short message to the referee. The referee then aggregates the messages and outputs the answer. Formally, each player is tasked with $V_i$ where the whole set of vertices in $G$ is $V = \sqcup_{i=1}^{k} V_i$. The samples of processor $i$ are then $\set{S(v)}_{v\in V_i}$. Each processor can send a message $a_i$ which is a function of its samples, and a referee receives all messages $a_1,\dots,a_k$ and outputs the answer.
The simultaneous first appeared in the context of testing independently in~\cite{FMO18} and preliminary version of~\cite{ACT20A,ACT20B}, where in the first $\card{a_i} = 1$ meaning each player is allowed to send one bit, and in the latter $\card{V_i} = 1$ meaning each player gets exactly one sample. We take the same point of view as in~\cite{FMO18}, but we remove the restriction of $1$ bit and allow a longer (but still short) message instead.

The number of players $k$ is given as a parameter, and our goal is to minimize the number of samples \emph{per player}, where all players get the same amount of samples: $q/k$ (we think of it as sort of parallelization of the sampling process).
We also consider the asymmetric-cost variant, where each player has an individual cost for each sample it draws (we think of this cost as the time it takes to draw each sample), and we wish to minimize the cost (or time) of the entire sampling process.

We also discuss the case of memory-constrained processors, also referred to as the streaming model, where this distribution testing was recently considered in~\cite{DGKR19}.

\subparagraph*{Memory-constrained processor.}
In the memory-constrained model each processor receives its samples as a stream, and once a sample is dealt with it is gone forever (this is the one-pass variant of the streaming model). A processor can only use a limited amount at each given moment, denote by $m$ (and measured by memory bits). We think of $m' = \lfloor m/(2\log n )\rfloor$ as the number of samples we can store with half the memory (we leave the other half for other calculations). The complexity measure of this model is the number of samples needed to complete the testing task.

Next we consider a new model that poses both constraints. It is a simultaneous model where each processor is memory-constrained, receiving its samples as a stream, and using its $m$ bits of memory it needs to come up with a message $a_i$ to send to the referee once all samples are seen. The referee then receives all messages and outputs an answer. We stick to the case where all processors are of the same type and therefore have the same constraints of $m$ bits.

\subparagraph*{The CONGEST Model.}
A standard model in the field of distributed computing is the CONGEST, where we have a communication network with $k$ players, and the networks runs a protocol that halts once a task is done. The complexity measure for this model is the number of communication round, and the focus is on congestions, as each communication edge can only transfer a small amount of bits at each round. In~\cite{FMO18} the task of testing in this model was considered and formulated in the following way: We have $k$ players over some communication graph, and each player is given one sample from an unknown distribution $P\in\Delta([n])$.
The players wish to run a communication protocol that ends when some player knows (with high probability) whether $P$ has some property, or is $\eps$-far from it.
We assume here that $k =\Omega(n/\eps^2)$ (as otherwise, the task becomes impossible), and we wish to minimize the amount of communication rounds, where at each round only $\Theta\left( \log n + \log k \right)$ bits can be communicated over each communication edge.

\section{Results}
\label{sec:results}
In this section we go over numerous applications of our method.
For each model we go over the same phases: we start with intuition as to which comparison graph $G$ is fit to this model, and we go on to show how one can simulate a collision-based algorithm $\alg{A} = (G,\tau)$. By simulating the algorithm we mean that by the end of the calculation, some processor will be able to compute both the number of collisions ($Z$) and the fitting threshold ($T$), so it is able to output the answer.
The last part is optimizing parameters, where first order parameters are those of $G$, and in some application we also give focus to second order parameters (choosing $\tau$).

The strength of the method comes from the following structural theorem that gives sufficient conditions for a comparison graph inducing a good uniformity tester:
\begin{theorem}\label{thm:structural}
    Fix a domain size $n$ and a proximity parameter $\eps$. If the following hold for an algorithm $\alg{A} := (G,\tau)$:
    \begin{enumerate}
        \item $\card{E} \geq \frac{4n}{\tau^2 \cdot \eps^4}$,

        \item $\card{E} \geq \frac{16n}{(1-\tau)^2 \cdot \eps^4}$, and

        \item $\frac{c(G)}{\card{E}^2} \leq \frac{(1-\tau)^2\eps^2}{16\sqrt{n}}$,
    \end{enumerate}

then $\alg{A}$ is an $\eps$-uniformity tester.

\end{theorem}

The proof, for the most part, is a generalization of the one used in~\cite{DGPP16} to show the original collision-based tester is optimal (in our notations, the original tester over $q$ samples is simply the algorithm $\alg{A} = (K_q,1/2)$).
While generalizing the proof we leave not one, but \emph{three} separate conditions on a general comparison graph, that together guarantee it induces an $(n,\eps)$-uniformity tester. This supplies a better, multi-dimensional understanding of how well a collision-based algorithm is guaranteed to perform.
We note that if one is willing to ignore constants, one could simply fix $\tau = 1/2$ and merge the first two conditions into one. However, interestingly for our method other values of $\tau$ (usually smaller) guarantee slightly better constants. We leave all $3$ conditions separate to maintain maximum flexibility when proving application of the theorem.

A specific comparison graph that is key to our algorithms is the one of \emph{disjoint cliques}. Indeed, our strategy will be ``perform any comparison you can'' which sometimes simply means we have several bulks of samples, where in each bulk all pairs can be compared. To this end, we give a more specific version of Theorem~\ref{thm:structural}, for graphs that have this structure:

\begin{lemma}\label{lem:disj_cliques}
    Fix $n, \eps$. Fix a comparison graph $G = \bigsqcup_{i=1}^{\ell} G_i$, where each $G_i$ is isomorphic to $K_{q}$, for some $q\geq 3$.
    If the following hold for an algorithm $\alg{A} := (G,\tau)$:
     \begin{enumerate}
            \item $q\sqrt{\ell} \geq \frac{\sqrt{12}\sqrt{n}}{\tau \cdot \eps^2}$

            \item $q\sqrt{\ell} \geq \frac{\sqrt{48}\sqrt{n}}{(1-\tau) \cdot \eps^2}$

            \item $q\ell \geq \frac{24\sqrt{n}}{(1-\tau)^2 \eps^2}$
     \end{enumerate}

    then $\alg{A}$ is an $\eps$-uniformity tester that uses $\card{V}$ samples.

\end{lemma}

Here again we leave the $3$-conditions version in order to be able to adjust the threshold parameter for optimizations. However, here all $3$ conditions are quite similar, pointing to the following simple corollary:
\begin{corollary}\label{cor:disj_cliques}
   Fix $n, \eps$. Fix a comparison graph $G = \bigsqcup_{i=1}^{\ell} G_i$, where each $G_i$ is isomorphic to $K_{q}$, for some $q\geq 3$. For each constant $\tau$, the algorithm $\alg{A} := (G,\tau)$ is an $\eps$-uniformity tester, if it holds that
   \[
        q\sqrt{\ell} \geq 35\sqrt{n} / \eps^2
   \]
\end{corollary}
\begin{proof}
  As $\ell \geq 1$, we have $q\ell \geq q\sqrt{\ell}$, and so the third condition can be relaxed in that manner. By fixing a constant value $\tau = 1/9$, all $3$ conditions of the lemma boils to the same asymptotic term with constant less than $35$ (and a slightly better guarantee can be made by optimizing over $\tau$).
\end{proof}

We are now ready to devise testers in various models.

\subsection{Standard Processors}
\label{subsec:res_standard_processors}

\subsubsection{Centralized Model}
As a warm up, we re-prove the original collision-based tester works, in term of the comparison graph, and using \ref{lem:disj_cliques}. To add a small twist, we show that under our analysis, better sample complexity is guaranteed when using a biased threshold (meaning, taking $\tau \neq 1/2$).
Let us denote $q:= \card{V}$, which is the number of samples drawn, and our complexity measure for the model. The following is rather straightforward:

\begin{corollary}\label{cor:centralized}
    One can test uniformity using $q = \Theta\left(\frac{\sqrt{n}}{\eps^2}\right)$.
\end{corollary}

\begin{proof}
  We simply use the lemma for $\ell = 1$ (one big clique), and then sufficient conditions on $q$ are:
  \begin{enumerate}
    \item $q \geq \frac{\sqrt{12}\sqrt{n}}{\tau \cdot \eps^2}$

    \item $q \geq \frac{\sqrt{48}\sqrt{n}}{(1-\tau) \cdot \eps^2}$

    \item $q \geq \frac{24\sqrt{n}}{(1-\tau)^2 \eps^2}$
  \end{enumerate}
All of which are fulfilled by choosing e.g., $q = \frac{100\sqrt{n}}{\eps^2}$ with $\tau = 1/2$. We note that the tester can easily compute $Z, T$ and therefore execute the collision-based algorithm $(G,\tau)$, for any value of $\tau$.

An added perk here, is that one can optimize $\tau$ over the three conditions to reduce sample complexity by a constant factor. Even though the guaranteed constant is somewhat of an artifact of the proof, it is still interesting to see that $\tau = 1/9$ would reduce the constant from $100$ to roughly $35$ (while simple optimization over $\tau$ would reduce it even a bit more, for some irrational threshold value).
\end{proof}

\subsubsection{Simultaneous Model}
In this section we discuss the simultaneous model, where $k$ players are each given oracle access to the distribution $P$. After taking samples, each player is allowed to send a short message to a referee, who then needs to output the right classification for $P$ (with high probability).

We give our focus to the variant posed in~\cite{FMO18}: what is the number of \emph{samples per player} needed to test uniformity?
The ``single collision'' algorithm devised for that question only requires \emph{one bit} from each player, and indeed it was shown to be optimal in~\cite{MMO19}.
However, this algorithm is somewhat delicate: first, the range of the parameter $k$ is limited (it cannot be too high or too low, with regards to $n, \eps$); second, if the number of players $k$ is not accurately known to all players, the algorithm breaks. This is because unlike most results in the classic setting, the analysis here actually requires each player to use a specific amount of samples, but not more or less. Because the players can only communicate a single bit, everything else must be set in advance given the problem's parameters.

To that end, we relax the model, and allow each player to send a small number of bits. These would allow each player to send the number of collisions she saw (rather than \emph{whether} a collision occurred). Our algorithm works for any parameter $k$, and can adjust to the case where each player is not exposed to the exact value of $k$, but rather to an approximation of it.

The model at hand imposes very concrete limitations on our comparison graph: one cannot compare a sample from one process to a sample of the another process. This means having $k$ players solving the problem in a parallel way, is equivalent to having a comparison graph whose number of connected components is \emph{at least} $k$. Followed by the intuition of ``compare every pair you have'', the graph we use is the disjoint cliques graph, with $k$ disjoint cliques. Our measure of complexity is the number of samples each player used, $q'$, which is translated to be the size of each one of our cliques.

\begin{corollary}\label{cor:simultaneous}
    In the simultaneous model, one can test uniformity using $q' = \Theta\left(\frac{\sqrt{n}}{\sqrt{k}\eps^2}\right)$ samples per player, where each player is allowed to send $\Theta\left( \log \left(1/\eps\right) \right)$ bits to the referee.
\end{corollary}

\begin{proof}
    We again turn to the structured graph of Lemma~\ref{lem:disj_cliques}, where each player uses an independent sample set of size $q'$ and compares all the pairs. We have $k$ players in total doing so.

    Our first goal is to show how to simulate a collision-based tester $(G,\tau)$ in this model. This will be possible whenever $G$ is made of at least $k$ vertex-disjoint connected components $G = \bigsqcup_{i\in[k]} G_i$.

    The first simple observation is that each player can operate on a designated part of the graph $G_i = (V_i, E_i)$. Meaning, the $i^{th}$ player is able to compute $Z_i:= \sum_{e\in E_i} \ind{e}$ by using her partial sample set, $V_i$.

    We next note that $Z = \sum_i Z_i$, and so if every player communicates $Z_i$ to the referee ``well enough'' -- she can go on to compute $Z$. Computing $T$ is easy: the algorithm is known to all, and specifically the values $\tau, \card{E}$ are known to the referee. The last part is comparing the values and output the decision $\ind{Z < T}$.

    Our next step is to quantify the communication and sampling cost of said algorithm. We aim at a communication cost of $\log\left( 1/\eps \right)$ bits per player.
    We would like each player to send the value $Z_i$. However, this value might be too high, especially for the case of distributions with very high collision probability (which are also very far from being uniform). To this end, we present the following simple observation:
    \begin{observation}
      To simulate a collision based algorithm $(G, \tau)$ in the simultaneous case, it suffices for each player to send $\log\left( T \right)$ bits. Indeed, whenever $Z_i > T$, we know that $Z > T$ as well, and a special string can signal the referee to reject.
    \end{observation}

    We now calculate the sample complexity, measured in samples \emph{per player}. We use Lemma \ref{lem:disj_cliques} with $q'$ and $\ell = k$, to know that it is enough if we satisfy:
    \begin{enumerate}
        \item $q' \geq \frac{\sqrt{12}\sqrt{n}}{\tau \cdot \sqrt{k} \eps^2}$

        \item $q' \geq \frac{\sqrt{48}\sqrt{n}}{(1-\tau) \cdot \sqrt{k}\eps^2}$

        \item $q' \geq \frac{24\sqrt{n}}{(1-\tau)^2 k \eps^2}$
    \end{enumerate}

    Which is enough to show the asymptotic sample complexity we desire.\footnote{Assuming large enough $k$, the third condition pose asymptotically weaker requirement, and so optimizing $\tau$ on the first two would yield the optimal guarantee for $\tau = 1/3$, getting us to $q' = \left\lceil \frac{\sqrt{108}\sqrt{n}}{\sqrt{k} \eps^2} \right\rceil$ samples per player.}

    Since the third condition is looser than the first two, we actually choose the number of samples per player, $q'$,  such that $\card{E} = \Theta \left( n/\eps^4 \right)$. Since $\tau\eps^2 \leq 1$, this means that our threshold is not too large
    \[
        T = \Theta\left( \frac{n}{\eps^4} \right) \cdot \left(\frac{1+\tau\eps^2}{n}\right) = \Theta\left( 1/\eps^4 \right),
    \]
    and so, using the observation, the simulation of the algorithm only requires each player to send $\log\left(\Theta(1/\eps^4)\right)$ bits to describe the number of collisions she saw.
\end{proof}

\begin{remark}
    As apparent from~\cite{ACT20A,ACT20B,MMO19}, when messages are $r$ bits long, the correct value to look for in the sample complexity is $2^r$ (and sometimes $\sqrt{2^r}$). For this reason, when discussing communication, we explicitly write expressions such as $\log\left(\Theta(1/\eps^4)\right)$ instead of a more general $\Theta\left(\log(1/\eps)\right)$. This would prove useful when comparing our results to known lower bounds, as is done in Section~\ref{subsec:discussion}.
\end{remark}

\begin{remark}
    For the oblivious case where the players only know an approximation for the number of players, $k$, each player would also need to send the referee the amount of comparisons she has made.
    We note that sending the amount of samples is enough, but still doing so naively raises the communication cost from $\log\left(\Theta(1/\eps^4)\right)$ to $\log\left(\Theta(n/\eps^4)\right)$.

    To avoid this, we suggest an alternative: each player rounds up her number of samples to the nearest power of $2$, and sends the power instead, so the referee can derive the exact number of samples (and therefore comparisons). For the number of collisions we could not have used the same trick, as it must be sent \emph{accurately}. To conclude, in the oblivious case, we can still perform uniformity testing with communication $\log\left(\Theta(1/\eps^4)\right) + \log\log\left(\Theta(n/\eps^4)\right)$ per player, while increasing the amount of samples (of each player) by a factor of at most $2$ (rounding to the nearest power).
\end{remark}

\subsubsection{Asymmetric-Cost Model}
This variant is best described and motivated as follows: we think of the sampling process as time consuming, and each of the $k$ players now has her own sampling rate, meaning that some players are able to draw samples faster than others. We describe the sampling rate vector $R = (R_1, \dots, R_k)$, and we let each player draw her own number of samples $(s_1,\dots, s_k)$. The complexity measure is the time dedicated for the sampling process, denoted $t$. Within $t$ time, player $i$ collects exactly $q_i:= \lfloor R_i \cdot t \rfloor$ samples.

Executing the tester works just as before. If the referee knows $(R_1,\dots,R_k)$, and the time $t$, she can calculate the threshold correctly, and as before - $Z = \sum_i Z_i$, where $Z_i$ is communicated by player $i$. The big difference here is the analysing the complexity measure, $t$. Note that the algorithm now assigns more responsibility to players who can obtain more samples within the same period of time.

To this end, we tweak the disjoint cliques graph in a way that each clique has a different size, relating to the sampling rate of said player. This would lead us to the following corollary:

\begin{corollary}\label{cor:simultaneous_asymmetric}
    In the simultaneous model, one can test uniformity in time $t = \Theta\left(\frac{\sqrt{n}}{\eps^2\lnorm{2}{R}}\right)$, where $\lnorm{2}{R} = \sqrt{R_1^2 + \dots + R_k^2}$, and each player is allowed to send $\Theta\left( \log \left(1/\eps\right) \right)$ bits to the referee.
\end{corollary}

\begin{proof}
    Our graph is now $G = \bigsqcup_{i\in[k]} G_i$, where $G_i$ is a clique on $q_i = R_i \cdot t$ vertices.
    We use the same observation as before, that no player needs to send more than $\log\left( T\right)$, where $T$ is the threshold of the algorithm $(G, \tau)$.

    In terms of the whole graph $G$, we have: $\card{V} = \sum_i q_i$. We have the bound
    \[
        \card{E} = \sum_i \frac{q_i^2-q_i}{2} \geq \frac{\sum_i (q_i ^2)}{3} = \frac{t^2 \cdot \lnorm{2}{R}^2}{3},
    \]
     as well as the bound
     \[
        c(G) = \sum_i \binom{q_i}{3} \leq \frac{\sum_i q_i^3}{6} = \frac{t^3 \cdot \lnorm{3}{R}^3}{6} \leq \frac{t^3 \cdot \lnorm{2}{R}^3}{6}.
     \]
     Combining the two we have
     \[
        \frac{\card{E}^2}{c(G)} \geq \frac{t^4 \lnorm{2}{R}^4 / 9}{t^3 \lnorm{2}{R}^3 / 6} = \frac{2 t \lnorm{2}{R}}{3},
     \]
     and so sufficient conditions for Theorem \ref{thm:structural} are:
    \begin{enumerate}
        \item $t \geq \frac{\sqrt{12}\sqrt{n}}{\tau \cdot \lnorm{2}{R} \eps^2}$

        \item $t \geq \frac{\sqrt{48}\sqrt{n}}{(1-\tau) \cdot \lnorm{2}{R}\eps^2}$

        \item $t \geq \frac{24\sqrt{n}}{(1-\tau)^2 \lnorm{2}{R} \eps^2}$
    \end{enumerate}
    Which is enough to show the asymptotic number we desire. Again, the algorithm uses $\card{E} = \Theta \left( n/\eps^4 \right)$ comparisons, and satisfy $T = O\left( 1/\eps^4 \right)$, which means each player needs to send at most $\log\left(\Theta(1/\eps^4)\right)$ bits to the referee.

\end{proof}
    This algorithm is actually a generalization of the previous two algorithms: (1) if the rate vector is $R = (1, \dots, 1)$, then $t$ is the number of samples per player, and $\lnorm{2}{R} = \sqrt{k}$; (2) if the rate vector is $R = (1, 0, \dots, 0)$, then only one player is actually taking samples, and indeed $t \cdot \lnorm{2}{R} = t$ is the number of samples.


\subsection{Memory-Constrained Processors}
\label{subsec:res_streaming_processor}
\subsubsection{Centralized Model with Memory Constraints}
In the streaming version of the problem, introduced in~\cite{DGKR19}, we get the samples from distribution $P$ in a stream, and can only store $m$ bits of memory. We wish to stream as few samples as possible and still output with high probability whether $P$ is uniform or $\eps$-far.
For our purposes, it will be useful to keep in mind the parameter of $m' = \left\lfloor m/2(\log(n)) \right\rfloor$, the number of \emph{samples} that can be stored with half of the memory space (the other half will be used to perform required calculations).

The restriction this model imposes is that when a certain sample is being processed, it can only be compared with the $O(m')$ samples that are currently stored. This intuition can be formalized to show that any comparison graph used in the streaming model must have a bounded \emph{average} degree (the maximal degree can be arbitrarily high, as one sample can be stored aside and compared with all the rest). Indeed, we show this later on (see Claim~\ref{claim:streaming_avg_deg}, in Section~\ref{subsec:lower_streaming_processors}) when dealing with lower bounds for our method. However, for now we merely use this intuition to choose a fitting comparison graph for the model.

In~\cite{DGKR19} the bipartite tester was used, with one side much smaller than the other. This graph goes along with the intuition, as it has a small average degree. The bipartite tester was shown to achieve sample complexity of $\Theta\left( n / (m'\eps^4)\right)$, which is shown to be optimal for a specific range of values for the parameter $m$. It was left as an open question whether this sample complexity can be attained for other values of $m$. For the upper bound, we answer this question in the positive. We use not only a different analysis, but rather a whole different comparison graph, one that also has a low average degree.
For simplicity, we again turn to the disjoint cliques graph, having dealt with similar analysis in the previous section (for the simultaneous model). However, we point out the fact that another graph leads to slightly better guarantees (by a small constant factor), with more cumbersome analysis.\footnote{The better comparison graph has the vertex set $V = \set{1,\dots,q}$, and the slightly larger edge set $E = \set{(i,j) \st \card{i-j} \leq m'-1}$. The tester it induces stores $m'$ samples at any given moment. For any new element, it replaces it with the oldest one in store ($m'$ samples ago) and compares it to the other $m'-1$ samples. Before this element is being removed, it is compared with $m'-1$ new samples that followed it. So other than the first and last $m'$ samples, each is being compared with $2m'-2$ other samples.}

\begin{corollary}\label{cor:streaming}
    In the streaming model with memory $m$ , one can test uniformity using
    \[
        q = \Theta\left(\max\set{\frac{n\cdot\log(n)}{m\eps^4}, \sqrt{n}/\eps^2} \right)
    \]
    samples, as long as we have $m = \Omega \left( \log (1/\eps)\right)$.
\end{corollary}

\begin{proof}
    Whenever we have enough memory $m = \Omega\left(\sqrt{n}\log(n) / \eps^2 \right)$, the dominating factor in the statement is the second one, and indeed one can store the whole $\sqrt{n}/\eps^2$ samples (using $\log(n)$ space per sample), and run the centralized algorithm.
    Whenever $m = o\left(\sqrt{n}\log(n) / \eps^2 \right)$, the memory limitation prevents us from running the centralized algorithm. For this case we choose a different comparison graph, inducing a tester that uses more samples: $\Theta\left(n\log(n)/(m\eps^4)\right)$.

    We start by describing how to execute the collision-based tester within the limitations of the model: we use half of the memory ($m/2$ bits) to maintain a global collision counter, and the other half would allow us to store $m'$ samples at each given time, so we are able to compare all $m'$ samples with one another. For each batch we store all $m'$ samples, compare them among themselves and add to global counter, and then delete them and clear the space for a new batch. We note that the global counter must suffice to compute $Z$ (for the right comparison graph $G$), and $T$ can be easily computed beforehand (knowing $m, \tau$, and $s$ -- the number of samples to be used). This is indeed enough to execute the collision-based algorithm $(G,\tau)$.

    As with the previous model, we again use a simple observation regarding the model:
     \begin{observation}
      In order to execute a collision-based algorithm $(G, \tau)$ in the streaming model, it suffices to allocate $\log\left( T \right)$ bits of memory for the global collision counter. This is true as once the global collision counter reaches the value $T$, the tester can immediately terminate and reject.
    \end{observation}

    We now need to choose the correct size of each batches. We use the disjoint cliques graph, this time with parameter $m'$ for the size of each clique.
    The total number of samples (and our complexity measure) is $q = \card{V}$, and we write for the number of cliques $\ell = q / m'$ (for simplicity we assume $\ell$ is an integer). Using Corollary~\ref{cor:disj_cliques}, with the values of $m', \ell$, our condition is:
    \[
        m' \sqrt{q/m'} \geq \frac{35\sqrt{n}}{\eps^2},
    \]
    and choosing $q, m'$ to make it an equality, after simplifying we get that
    \[
        q = \frac{3000n \log(n)}{m \eps^4}
    \]
    are enough samples.

    Note that the tester is defined only when $1 \leq \ell = q / m' = 2\log(n)\sqrt{3000n}/\left( m \eps^2\right)$, which is equivalent to $m \leq 2\log(n)\sqrt{3000n}/\eps^2$. This complies with the fact that larger $m$ already allows us to use the centralized tester.

    Taking the minimal required number of samples gives us a comparison graph with $\card{E} = \Theta \left( n/\eps^4\right)$, which in turn dictates a threshold $T = \Theta\left( 1/\eps^4 \right)$.
    Using the observation, we must have $m/2 \geq \log\left( T \right)$, which translates to $m = \Omega \left( \log (1/\eps)\right)$.
    We also require that $m' \geq 3$ (a prerequisite of Corollary~\ref{cor:disj_cliques}, even though we can also make it work for $m' = 2$), overall implying correctness whenever $m \geq \Theta\left(\max\set{\log (1/\eps), \log n}\right)$.
\end{proof}

The corollary shows a tester for a wide range of $m$ (extending the range achieved in~\cite{DGKR19}). We are still left with the restriction of $m = \Omega\left(\max\set{\log (1/\eps), \log n}\right)$, which essentially comes from our need to count collisions and to be able to store at least $2$ full samples at the same time.

We leave open an interesting question posed also in~\cite{DGKR19}: is it possible to test for uniformity in the scarce regime, and if so -- what is the sample complexity required to do so? (The scarce regime is essentially $m = o\left(\log(n)\right)$. Though, to comply with our result - e.g., for absurdly small values of $\eps$ - it can be redefined to $m = o\left(\max\set{\log (1/\eps), \log n}\right)$ )

\subsubsection{Simultaneous Model with Memory Constraints}
The main difficulty in this section seems to be handling the multiple parameters: on top of $n, \eps$ that define the testing problem, we now deal with $k$ processors only able to communicate a short message to the referee. Each of which can store at most $m$ bits of memory (or $m'$ samples, as before) in order to process its inputs and compute its message.

\begin{corollary}\label{cor:simultaneous streams}
    In the memory-constrained simultaneous model, whenever $m = \Omega\left(\log(1/\eps)\right)$ bits of memory are allowed for each player, and $\log\left(1/\eps^4\right)$ bits of communication are allowed for each player, uniformity can be tested using
    \[
        \Theta\left(\frac{n}{k \eps^4} \cdot \max \set{\frac{\log(n)}{m},\frac{\sqrt{k}\eps^2}{\sqrt{n}}}\right)
    \]
    samples per player.
\end{corollary}

    Note that in the regime of very small memory constraint, the problem suddenly "parallelize" perfectly: imagine taking $k \leq \Theta\left((n/\eps^4) \cdot (\log^2(n) / m^2)\right)$ machines with memory $m$ per machine. In this case, each machine needs only a $1/k$ fraction of the amount of samples required for a single machine with memory $m$.
    This is true as the real barrier here is the overall storage, which indeed grows linearly in the number of machines used.

\begin{proof}
    We first note that if $m' \geq c\sqrt{n/(k\eps^4)}$, for the constant $c$ of the prove in Corollary~\ref{cor:simultaneous}, then each player can store all the samples she needs, and therefore $\Theta\left(\sqrt{n/(k\eps^4)}\right)$ samples per player are enough.

    Next, we show the required adjustments whenever $m' < c\sqrt{n/(k\eps^4)}$.\footnote{Whenever $m' = \Theta(\sqrt{n/(k\eps^4)})$, these adjustments are required, but the asymptotic sample complexity is similar to the case of having enough memory. In this case, the two terms in the maximum function have the same asymptotic value.}

    Under this regime, the number of samples we are able to store limits us. As we have $k$ simultaneous players, our graph can be written as $G = \bigsqcup_i G_i$, where $G_i$ is the graph of the player $i$. Only now, each $G_i$ should be adjusted to the memory constraint, and this can be done as before, by dividing each $G_i$ to disjoint cliques, meaning $G_i = \bigsqcup_j G_{ij}$, where each $G_{ij}$ is isomorphic to $K_{m'}$.

    The execution of $(G,\tau)$ then can be done, relying on previous observations: each players uses half her space $m/2$ to calculate the value $Z_i$ (or a special signal for the case $Z_i > T$). She then sends these bits to the referee, using $\log(T)$ bits of communication. This can be done as long as  $m = \Omega\left(\log(T)\right)$.

    To calculate the sample complexity (per player), we use $q=\card{V}$ for the total amount of samples and note that our entire graph is a set of $\ell = q/m'$ disjoint cliques of size $m'$ each. Using Corollary~\ref{cor:disj_cliques}, we know the algorithm successfully tests for uniformity whenever
    \[
        m' \sqrt{q/m'} = \Omega\left( \sqrt{n}/\eps^2 \right) .
    \]
    And after simplifying, and choosing the minimal $q$, we get a tester with
    \[
        q = \Theta\left( \frac{n \log(n)}{m \eps^4} \right) .
    \]
    We do note, however, that $q$ is the total amount of samples, not our complexity measure. But using symmetry, it means that each player uses
    \[
        \Theta\left( \frac{n \log(n)}{m k \eps^4} \right)
    \]
    samples.

    We finish by noting again that in our comparison graph we have $\card{E} = \Theta\left( n/\eps^4\right)$, which means $T = O\left( 1/\eps \right)$, and so the communication and memory constraints are as in previous sections: we must have $m = \Omega \left( \log (1/\eps)\right)$ bit of memory per player, and each player sends only $O\left( \log (1/\eps)\right)$ bits to the referee.
\end{proof}

\subsection{The CONGEST Model}
\label{subsec:res_congest}
\subsubsection{Problem Definition}
In the CONGEST model, players can communicate over private peer-to-peer channels (communication edges), and the process is divided to rounds.
Each round allows $O\left(\log n+\log k\right)$ bits of communication on each channel (intuitively, allowing to send a constant amount of input tokens, or node identifiers).
The complexity measure is the number of communication rounds needed to solve a problem, as the algorithm tries to avoid congestions (the case where a large amount of information is destined to go through a specific communication edge).

The output of model is YES if no player raised a flag, and NO is at least one player did so. This can be defined similarly for the testing problem. However, using extra $D$ rounds (where $D$ is the diameter of the communication graph) - the task is equivalent to having one player output the answer. As we do not aim for less than $O(D)$ rounds, we relax and present the following CONGEST version of uniformity testing (which is similar to the one in~\cite{FMO18}):

We have a communication graph with $k$ processors, each holds exactly one sample from the distribution $P$. We only consider the case of having enough samples throughout the whole graph (meaning $k = \Omega(\sqrt{n}/\eps^2)$, as otherwise the task is impossible). An $\eps$-uniformity tester is a communication procedure that ends where one player outputs a single bit $b$, such that:
\begin{enumerate}
  \item If $P = U_n$, then $b = 1$ with probability at least $2/3$.

  \item If $P$ is $\eps$-far from uniform, then $b = 0$ with probability at least $2/3$.
  \end{enumerate}
The complexity of the algorithm is measured by the number of communication rounds needed to channel enough information to one node.

\subsubsection{An Improved Algorithm}
In~\cite{FMO18}, a tester for this task was given with round complexity of $O(D + n/(k\eps^4)$ rounds. We revisit this algorithm and show that certain communication graphs are good for testing. In these graphs, the players can: (1)detect the graph is good within $O(D)$ rounds; (2)shave the factor of $O(n/(k\eps^4))$ rounds, by using a different testing routine, hence solve the testing problem within $O(D)$ rounds. Hence, without hurting the round complexity on a general graph, we get better round complexity whenever the graph is ``good'' (without having prior knowledge that this is the case).

We next give a sketch of the algorithm, and a proof of its correctness.
The idea behind this algorithm is gathering sets of samples in a small amount of ``important'' nodes that are well-connected. These nodes would then simulate virtual nodes, each having a decent-sized set of samples (one node might simulate more than one virtual node). We then let each virtual node run a testing procedure of the simultaneous case, and collect the answers to aggregate them at one need, acting as the referee of the simultaneous model.

The details are as follows:
\begin{enumerate}
  \item The players identify the player with the largest identifier, node $r$, and construct a BFS tree $T$ rooted at $r$. This can be done using $O(D)$ rounds.
  \item The players cunningly pipeline samples up the tree to collect them in bundles of $s$ samples: first they count the amount of samples in each subtree, and then they only pipeline the remainders, each player $i$ keeping $b_i\cdot s$ samples for himself, where $b_i$ is an integer -- the amount of bundles. The pipelining takes $O(D+s)$ rounds.
  \item Each player with $b_i > 0$ simulates for each of his $b_i$ bundles a virtual node running on this sample set. Each such simulation is done independently. This part takes no communication at all.
  \item The players propagate the answers of simulations up the tree (as the answers are cheap in communication), and once all of these have reached the root $r$, it simulates the referees and outputs the answer. This takes another $O(D)$ rounds.
\end{enumerate}

We note that optimizing over the bundle size $s$ gives us $s = n/(k\eps^4)$, which can potentially be much larger than $D$, in which case the bottleneck of the algorithm is pipelining the samples up the tree.

We now turn to think of the communication graph as a potential comparison graph. Indeed, each node starts with a single sample, and comparing two samples of neighboring nodes takes a single round. Thus, it is a local process, which is very cheap in round complexity, and use communication over all channels. To this end, we suggest the following two procedures:
\subparagraph*{Detection procedure.}
In order to detect that our communication graph is a good comparison graph, one needs to have the values $\card{E}, c(G)$. We use $d_i$ for the degree of node $i$. We then have:
\begin{align*}
  \card{E} &= (1/2) \sum_i d_i &
  c(G) &= \sum_i \binom{d_i}{2}
\end{align*}
These sizes are also bounded by $k^2$ and $k^3$ (where $k = \card{V}$ is both the number of nodes, and the number of samples). This means that we can use messages of $O\left(\log(k)\right)$ bits to perform the summation up the tree, and using $O(D)$ rounds of communication, the root $r$ knows both these values, as well as $n, \eps$, and it can determine whether the communication graph works for some threshold value $\tau_0$. well as a comparison graph. It takes another $O(D)$ to propagates the answer to all nodes.

\subparagraph*{Alternative testing.}
In the case of a positive detection, we use a \emph{single} round for each player to send its sample to every neighbor with a larger identifier.\footnote{This is not crucial. If we do not have order on the vertices, they can each send all of their neighbors, and each comparison will be done exactly twice}
Next, each player locally counts the number of collisions it had with its neighbors $Z_i$ (no communication needed), and we note that the total amount of collisions is $Z = \sum_i Z_i$.
As $Z$ is bounded by the number of comparisons (edges), we know that all partial sums of values $Z_i$ can be expressed using $O\left(\log(k)\right)$ bits. Thus, we can use an extra $O(D)$ rounds to sum the number of collisions up the tree.
As the root can also calculate $T$, we actually get a simulation of the collision tester $\alg{A} = (G,\tau_0)$, where $G$ is the communication graph.

To sum up, both the detection, and the alternative testing can be done in $O(D)$ rounds, which means that for any network that has a good topology (depending on $n, \eps$ and the choice of $\tau$) -- the players can solve uniformity using $O(D)$ rounds.

\subsubsection{Discussion}
The improvement above shows how to connect an interactive model to the framework of collision-testers in a non-trivial way.
However, in this implementation only extremely local comparisons were made: ones that takes $1$ round to perform.

One can generalize this idea using comparison of higher constant order, $t$ (namely, all samples that are of distance $t$ from one another). This would mean taking the graph $G^t$ as our comparison graph. This graph has more edges, and is therefore should be a better fit.\footnote{As it turns out, not every comparison necessarily helps, but usually adding comparisons does help}
Indeed, one should modify the detection process: we will need to calculate $\card{E(G^t)}, C(G^t)$ up the tree, and each node would have to make sure its $t$-environment is no too large, as to not create a local congestion (as such congestions can occur even in a constant-diameter neighborhood). If the detection process passes, the $t$-local comparisons can be made for any constant parameter $t$. We expect larger value of $t$ to work for a broader family of graphs (for example, whenever $G$ is merely connected, the graph $G^k$ is the $k$-sized clique, which is known to be a good tester for large enough $k$)

We conclude with this: one can use $O(D)$ rounds to try detecting the local testability of the graph $G^t$, for any constant $t$. Depending on further assumptions, one can conduct a smart search of a locality parameter $t$ (trying different values along the way), especially whenever $n/(k\eps^4)$ is much larger than the diameter $D$, and we wish to avoid the global pipelining procedure.

\subsection{Testing Identity to a Fixed Distribution}
For the classic model, it was shown in~\cite{DK16,G16} that uniformity testing is complete with respect to testing identity to a fixed distribution, $D$.
Specifically, the black-box reduction in~\cite{G16} uses a random filter $F_D$ which we can apply on our sample set to obtain a new one. This is done in a way that every sample from $D$ becomes a uniform sample (on a slightly larger domain, $[m]$), and similarly a sample from any distribution that is $\eps$-far from $D$, becomes a sample from some a distribution that is $\eps'$-far from uniform (on $[m]$).
The trick is the existence of such a filter with $m = \Theta(n)$ and $\eps' = \Theta(\eps)$, which roughly maintains the sample complexity.

We note that the same reduction applies for all of the results above. Indeed, in all of them every node (or sample) in $\alg{A} = (G,\tau)$ is known at first to a single processor (the one who draws the sample). Thus, all processors can run the same global filter $F_D$ that only depends on $D$ (we assume here they are know the goal of the protocol, and the distribution $D$ we are testing for).

Formally, each node $v$ is originally associated with a sample $S_P(v)$, and a processor that holds this sample. Running it through the filter, we get a fresh sample $F_D(S_D(v))$ on which we run the rest of testing procedure. Overall, the nodes of our graph $G$ are now associated with $\card{V}$ samples that are either from uniform distribution (on $[m]$) or on an $\eps'$-far distribution, and from here we can use a collision-based uniformity testing.



\section{Proof of the Structural Theorem}
We now prove our structural theorem, as well as the more specific lemma for the disjoint cliques graph.

\subsection{Proof of Theorem~\ref{thm:structural}}
\label{subsec:thm_pf}
The outline for the proof follows an outline similar to the one of \cite{DGPP16}. Our main challenge is to fit the proof for \emph{any} comparison graph, rather than a specific one -- that of a full clique -- which describes the classic centralized uniformity tester.

By doing so, we get not one, but \emph{three} separate conditions on the \emph{comparison graph} that together guarantee that it can be used to plan an $(n,\eps)$-uniformity tester. This supplies a better, multi-dimensional understanding of how well a collision-based algorithm is guaranteed to perform.
We note that if one is willing to ignore constants, one could simply fix $\tau = 1/2$ and combine the first two conditions into one. \emph{However}, a side goals of ours is to optimize the analysis for both known results and our new ones, and for that goal we leave $3$ different conditions and optimize over $\tau$ to show that putting the threshold right in the model does not give the best complexity guarantee. We are then left with the following formulation:

\begin{theorem-repeat}{thm:structural}
    Fix a domain size $n$ and a proximity parameter $\eps$. If the following hold for an algorithm $\alg{A} := (G,\tau)$:
    \begin{enumerate}
        \item $\card{E} \geq \frac{4n}{\tau^2 \cdot \eps^4}$,

        \item $\card{E} \geq \frac{16n}{(1-\tau)^2 \cdot \eps^4}$, and

        \item $\frac{c(G)}{\card{E}^2} \leq \frac{(1-\tau)^2\eps^2}{16\sqrt{n}}$,
    \end{enumerate}

    then $\alg{A}$ is an $\eps$-uniformity tester.

\end{theorem-repeat}

To prove our theorem, we first need to understand how the tester works. The goal of a collision-based tester is to estimate $\mu$ well enough to distinguish between the YES and NO cases. The reason such testers work is because for the uniform distribution $U_n$, we have $\mu_{_{U_n}} = \frac{1}{n}$, and for any $\eps$-far distribution $P$, one can easily show that $\mu_P \geq \frac{1+\eps^2}{n}$.
This means that $\EE[Z]$ is noticeably different when fed samples from a YES case and when samples are drawn from a NO case. Therefore, our first proposition is that:

\begin{proposition}\label{prop:expect}
    Fix an algorithm $\alg{A} = (G,\tau)$, and an input distribution $P\in\Delta([n])$. We have
    \[
        \EE[Z] = \card{E}\cdot\mu .
    \]
\end{proposition}

\begin{proof}
    For each comparison edge $(u,v) = e\in E$, we have $ \EE[\ind{e}] = \Pr_S[S(u) = S(v)] = \mu $, the probability of a collision when taking two samples. Thus:
    \[
        \EE[Z] = \EE[\sum_{e\in E} \ind{e}] = \sum_{e\in E} \EE[\ind{e}] = \sum_{e\in E} \mu = \card{E} \cdot \mu .
    \]
\end{proof}

As per usual, the next step is to bound the variance of the random variable $Z$ in order to later show concentration to some degree in both YES and NO cases.
The bound we get consists of two separate terms. The first term can be though of as an unavoidable, inherent variance that comes along with each and every comparison.
However, the second term is more tricky, and it is created and described by the dependencies between different comparisons.
For example, in the case where all comparisons are done on freshly drawn pairs of samples, all comparisons are completely independent and indeed the comparison graph -- a perfect matching -- has $c(G) = 0$ which cancels the second term completely.

\begin{lemma}\label{lem:var}
    Fix an algorithm $\alg{A} = (G,\tau)$, and an input distribution $P\in\Delta([n])$. then
    \[
        \Var[Z] = \card{E}\cdot (\mu - \mu^2) + c(G)\cdot (\gamma - \mu^2) .
    \]
\end{lemma}

\begin{proof}
    By definition, $Z = \sum_{e\in E} \ind{e}$, and so
    \[
        \EE[Z^2] = \EE\left[\left(\sum_{e\in E} \ind{e} \right)^2\right] = \EE\left[\sum_{\substack{e_1\in E \\ e_2\in E}} \ind{e_1}\cdot \ind{e_2} \right] = \sum_{\substack{e_1\in E \\ e_2\in E}} \EE[\ind{e_1} \cdot \ind{e_2}] .
    \]

    This sum consists of $\card{E}^2$ summands. We break these into 3 types of summands:
    \begin{itemize}
        \item $\card{e_1 \cup e_2} = 2$, which means $e_1 = e_2$. In this case the summand is simply
        \[
            \EE[\ind{e_1}^2] = \EE[\ind{e_1}] = \mu .
        \]

        \item $\card{e_1 \cup e_2} = 3$. i.e., we have a common vertex. We call the 3 vertices $u,v,w$, and calculate:
        \[
            \EE\left[\ind{e_1} \cdot \ind{e_2}\right] = \EE\left[\ind{S(u) = S(v) = S(w)} \right] = \gamma .
        \]

        \item $\card{e_1 \cup e_2} = 4$. In this case we have 4 distinct vertices, and so the two indicators $\ind{e_1}, \ind{e_2}$ are independent, and we have
        \[
            \EE[\ind{e_1} \cdot \ind{e_2}] = \EE[\ind{e_1}] \cdot \EE[\ind{e_2}] = \mu^2 .
        \]
    \end{itemize}
    There are exactly $\card{E}$ summands of the first type.
    Each summand of the second type corresponds uniquely to a directed $2$-paths in $G$, the one that consists of $e_1, e_2$. As our entire sum goes we over all (ordered) pairs of edges, we must include all directed $2$-paths, and so we have exactly $c(G)$ summands of this type.\footnote{As this might be confusing, we note that both pairs $e_1, e_2$ and $e_2,e_1$ appear in our sum, and indeed each one of them corresponds to a different directed $2$-path}
    The rest of the summands must be of the third type, and their number amounts to $\card{E}^2 - \card{E} - c(G)$.

    Putting it all together, we have
    \begin{align*}
    \EE[Z^2]
    &=
    \card{E}\cdot\mu + c(G)\cdot\gamma + (\card{E}^2 - \card{E} - c(G))\cdot \mu^2\\
    &=
    \card{E}^2\cdot\mu^2 + \card{E}\cdot(\mu-\mu^2) + c(G)\cdot(\gamma-\mu^2) .
    \end{align*}

    We now use Proposition\ref{prop:expect} to see that $\EE[Z]^2 = \card{E}^2\cdot \mu^2$. Plugging this in, we get:
    \[
        \Var[Z] = \EE[Z^2] - \EE[Z]^2 = \card{E}\cdot(\mu-\mu^2) + c(G)\cdot(\gamma-\mu^2) ,
    \]
    Which finishes the proof.
\end{proof}

It is interesting that to accurately express the expectation and variance of $Z$, it is enough to sum up the whole graph $G$ with only two quantified values, and even a bit surprising that $\card{V}$ is not even one of them (though, we later show in Section~\ref{subsec:graph_prop}, that there is an indirect connection).

We now wish to show that any collision-based tester that has certain properties can test uniformity well. After fixing the tester, we partition the proof into three cases according to the input distribution being tested, and show that in each case the tester indeed errs with probability at most $1/4$.

We first give a specific application of Chebyshev's inequality, one that will be used in \emph{all} three cases to follow:

\begin{lemma}\label{lem:chevyshev}
Fix an algorithm $\alg{A} = (G,\tau)$, and an input distribution $P\in\Delta([n])$. If
\begin{equation}\label{eq:chevyshev}
4\Var[Z] \leq \card{E}^2 \cdot \left(\frac{1+\tau\eps^2}{n} - \mu\right)^2
\end{equation}
then $\alg{A}$ errs on the distribution $P$ with probability at most $1/4$.
\end{lemma}

\begin{proof}
Using simpler notations, the r.h.s of \eqref{eq:chevyshev} equals to $\left(\card{E}\cdot \frac{1+\tau\eps^2}{n} - \card{E}\cdot\mu\right)^2 = (T - \EE[Z])^2$, which means \eqref{eq:chevyshev} can be written as:
\[
    4\Var[Z] \leq (T-\EE[Z])^2 .
\]
We also note that since $\tau\in[0,1]$, the threshold is right in between the value of $\EE[Z]$ for the YES case (which corresponds to $\tau = 0$), and the value of $\EE[Z]$ for \emph{any} NO case (which can be as low as the value corresponding to $\tau = 1$).

Now, looking at events over the probability space of the sampling process, one can see that every time the tester $\alg{A}$ errs, it means the value $Z$ ``fell'' on the wrong side of the threshold, which means $T$ lies in between $Z$ and $\EE[Z]$. Formally, it means
that whenever $\alg{A}$ fails, $\card{Z-\EE[Z]} \geq \card{T-\EE[Z]}$. So the last event's probability is larger then the former's:
\[
    \Pr[\alg{A} \text{ fails}] \ \ \leq \ \ \Pr\Big[\card{Z-\EE[Z]} \geq \card{T-\EE[Z]}\Big] \leq \frac{\Var[Z]}{(T-\EE[Z])^2} \leq \frac{1}{4} .
\]
Where for the second and third inequalities we used Chebyshev's inequality, and \eqref{eq:chevyshev}.
\end{proof}

Equipped with Lemma~\ref{lem:chevyshev}, we are ready to find a sufficient condition for the tester to succeed with high probability in each of the three cases. The first type will simply be the YES case where $P = U_n$.

\begin{lemma}[Uniform distribution]\label{lem:uni}
Fix an algorithm $\alg{A} = (G,\tau)$, and fix $P = U_n$. If
\begin{equation}\label{eq:cond_uni}
\card{E} \geq \frac{4(n-1)}{\tau^2\eps^4}
\end{equation}
then $\alg{A}$ fails on $P$ with probability at most $1/4$.
\end{lemma}

\begin{proof}
For the uniform distribution we have $\mu = 1/n$, and $\gamma = 1/n^2$. And so using \ref{lem:var}:
\[
    \Var[Z] = \card{E}\cdot (1/n - 1/n^2) + c(G)\cdot (1/n^2 - 1/n^2) = \card{E}\cdot \left(\frac{n-1}{n^2}\right) .
\]
Using Lemma~\ref{lem:chevyshev}, we see that for the algorithm to work well on the uniform distribution, it is enough to require:
\[
    4\card{E}\cdot \left(\frac{n-1}{n^2}\right) \leq \card{E}^2 \cdot \left(\frac{\tau\eps^2}{n}\right)^2 ,
\]
which after re-arranging is equivalent to equation \eqref{eq:cond_uni}.
\end{proof}

We now turn to the two other cases. We split the set of $\eps$-far distributions into two, according to the dominant summand in the variance of $Z$ after fixing a tester $\alg{A} = (G, \tau)$. Formally, we define
\[
    \mathcal{P}_1(\alg{A}) := \set{P \in \Delta([n]) \st  \norm{P - U_n} \geq \eps \text{  and  } \card{E}\cdot(\mu-\mu^2) \geq c(G)\cdot(\gamma-\mu^2)} ,
\]
\[
    \mathcal{P}_2(\alg{A}) := \set{P \in \Delta([n]) \st  \norm{P - U_n} \geq \eps \text{  and  } \card{E}\cdot(\mu-\mu^2) < c(G)\cdot(\gamma-\mu^2)} .
\]

The two ``types'' together cover every $\eps$-far distribution (no matter what $\alg{A}$ was chosen).
We first show that both summands in the variance are non-negative, since the two terms $\mu-\mu^2$ and $\gamma-\mu^2$ are non-negative for any distribution $P$.

Indeed, for the first term we use the common fact that $\frac{1}{n} \leq \mu \leq 1$, and so $\mu - \mu^2 = \mu(1-\mu) \geq 0$.

For the second, we use Cauchy-Schwarz inequality:
\[
    \mu^2 = \left(\sum_{i=1}^{n} P_i^2\right)^2 = \left(\sum_{i=1}^{n} P_i^{0.5} \cdot P_i^{1.5}\right)^2 \leq \sum_{i=1}^{n} P_i \cdot \sum_{i=1}^{n} P_i^3 = \gamma .
\]

Now, once a tester $\alg{A}$ is fixed, we have the two lemmas that follow, each deals with one of the two distributions families induced by the algorithm.

\begin{lemma}[Distributions from $\mathcal{P}_1(\alg{A})$]\label{lem:type1}
Fix an algorithm $\alg{A} = (G,\tau)$, and an $\eps$-far distribution $P\in\mathcal{P}_1(\alg{A})$.
If
\[
    \card{E} \geq \frac{16n}{(1-\tau)^2\eps^4} ,
\]
then $\alg{A}$ fails on $P$ with probability at most $1/4$.
\end{lemma}

\begin{proof}
Since $P\in\mathcal{P}_1(\alg{A})$, we know that $4\Var[Z] \leq 8\card{E}\cdot(\mu-\mu^2) \leq 8\card{E}\cdot\mu$.

Therefore, to satisfy~\eqref{eq:chevyshev}, it suffice to require
\[
    8\card{E}\cdot\mu \leq  \card{E}^2 \cdot \left(\frac{1+\tau\eps^2}{n} - \mu\right)^2 .
\]

Let $\alpha$ be such that $\mu = \frac{1+\alpha}{n}$. Since $P$ is $\eps$-far from uniform, we know that $\alpha \geq \eps^2$. Using this notation, we have
\[
    8\card{E}\cdot\left(\frac{1+\alpha}{n}\right) \leq  \card{E}^2 \cdot \left(\frac{1+\tau\eps^2}{n} - \frac{1+\alpha}{n}\right)^2 ,
\]
and after re-arranging, we get
\[
    \card{E} \geq \frac{8(1+\alpha)n}{(\tau\eps^2-\alpha)^2} .
\]

As this must be true for \emph{any} value of $\alpha$, we show that it is enough to require that it holds for $\alpha = \eps^2$.

In order to show this, we think of $f(\alpha) = 8n(1+\alpha)$, $\ \ g(\alpha) = \tau\eps^2-\alpha$, and so
\[
    \left(\frac{f}{g^2}\right)' = \frac{8n(1+\alpha)'\cdot g^2(\alpha) - 8n\cdot 2g(\alpha)}{g^4(\alpha)} = \frac{8n}{g^4(\alpha)} \cdot (1+\alpha -2g(\alpha)) \cdot g(\alpha) .
\]

But the term $8n / g^4(\alpha)$ is positive, and whenever $\alpha \geq \eps^2$ (and since $\tau \leq 1$), we also have it that $g(\alpha) \leq 0$ as well as $(1+\alpha-2g(\alpha) \geq 0$. To conclude, for the range $\alpha \geq \eps^2$, we have that $(f/g^2)' \leq 0$, which means the maximum of the function $(f/g^2)$ in exactly when $\alpha = \eps^2$.

A sufficient condition, then, for \eqref{eq:chevyshev} to hold (and subsequently, for $\alg{A}$ to err on $P$ with probability at most $1/4$) is
\[
    \card{E} \geq \frac{16n}{(1-\tau)^2\eps^4} ,
\]
which ends the proof.

\end{proof}

\begin{lemma}[Distributions from $\mathcal{P}_2(\alg{A})$]\label{lem:type2}
Fix an algorithm $\alg{A} = (G,\tau)$, and fix an $\eps$-far distribution $P\in\mathcal{P}_2(\alg{A})$. If
\[
    c(G) \leq \frac{\card{E}^2(1-\tau)^2\eps^2}{16\sqrt{n}} ,
\]
then $\alg{A}$ fails on $P$ with probability at most $1/4$.
\end{lemma}

\begin{proof}
We again want to show that given our premise, inequality \eqref{eq:chevyshev} holds.

We start with a finer analysis of $\gamma - \mu^2$. For this purpose, we define the $n$ dimensional vector $a := P - U_n$, which means $a_i = P_i - 1/n$. We thus have $\sum_{i=1}^{n} a_i = 0$. Now, we compute:
\begin{align*}
    \gamma - \mu^2
    &=
    \sum_{i=1}^{n} (1/n + a_i)^3 - \left[\sum_{i=1}^{n} (1/n + a_i)^2\right]^2\\
    &=
    \left(\sum_{i=1}^{n} 1/n^3 + 3\sum_{i=1}^{n} a_i/n^2 + 3\sum_{i=1}^{n} a_i^2/n + \sum_{i=1}^{n} a_i^3\right)\\
    &- \left(\sum_{i=1}^{n} 1/n^2 + 2\sum_{i=1}^{n} a_i/n + \sum_{i=1}^{n} a_i^2\right)^2\\
    &=
    \left(1/n^2 + 0 + (3/n)\norm{a}_2^2 + \norm{a}_3^3\right) - \left(1/n + 0 + \norm{a}_2^2\right)^2\\
    &=
    \left(1/n^2 + (3/n)\norm{a}_2^2 + \norm{a}_3^3\right) - \left(1/n^2 + (2/n)\norm{a}_2^2 + \norm{a}_2^4\right)\\
    &=
    \norm{a}_2^2/n + \norm{a}_3^3 - \norm{a}_2^4\\
    &\leq
    \norm{a}_2^2/n + \norm{a}_2^3\\
    \end{align*}
In the last step we used the fact $\norm{\cdot}_3 \leq \norm{\cdot}_2$, and omitted the last term. We now again denote $\mu = \frac{1+\alpha}{n}$, so that $\norm{a}_2^2 = \alpha/n$. Combining the above with Lemma \ref{lem:var}, and the assumption on $P$, we can upper-bound the variance:
\begin{equation}\label{eq:var_bound}
4\Var[Z] \leq 8c(G)\cdot(\gamma-\mu^2) \leq 8c(G)\cdot\left( \frac{\alpha}{n^2} + \frac{\alpha^{1/5}}{n^{1.5}} \right) = 8c(G) \cdot \left(\frac{1}{\alpha} + \frac{\sqrt{n}}{\sqrt{\alpha}}\right)\frac{\alpha^2}{n^2}
\end{equation}
The premise we have is equivalent to:
\[
    \frac{\card{E}^2}{c(G)} \geq \frac{16\sqrt{n}}{(1-\tau)^2\eps^2} = \frac{8}{(1-\tau)^2}\cdot \left(\frac{\sqrt{n}}{\eps^2}+\frac{\sqrt{n}}{\eps^2}\right) ,
\]

and using $\sqrt{n} \geq 1$ and $\eps^2 \leq \alpha$, as well as $\eps^2 \leq \eps \leq \sqrt{\alpha}$, we deduce:
\[
    \frac{\card{E}^2}{c(G)} \geq \frac{8}{(1-\tau)^2} \cdot \left(\frac{1}{\alpha} + \frac{\sqrt{n}}{\sqrt{\alpha}}\right) ,
\]
which, after re-arranging, is equivalent to
\[
    8c(G) \cdot \left(\frac{1}{\alpha} + \frac{\sqrt{n}}{\sqrt{\alpha}}\right) \leq \card{E}^2 (1-\tau)^2 .
\]
Combined with inequality \eqref{eq:var_bound} which bounds the variance, we have:
\[
    4\Var[Z] \leq \card{E}^2 \frac{(1-\tau)^2\alpha^2}{n^2} .
\]
We again use the fact $\alpha \geq \eps^2$ to deduce that $(1-\tau)\alpha = \alpha - \tau\alpha \leq \alpha - \tau\eps^2$. Now we can conclude:
\[
    4\Var[Z] \leq \card{E}^2 \left(\frac{(1-\tau)\alpha}{n}\right)^2 \leq \card{E}^2 \cdot \left(\frac{\alpha - \tau\eps^2}{n}\right)^2 = \card{E}^2 \cdot \left(\frac{1+\tau\eps^2}{n} - \mu\right)^2 ,
\]
Which means that indeed \eqref{eq:chevyshev} holds, and thus by applying \ref{lem:chevyshev}, the tester $\alg{A}$ errs on samples from $P$ with probability at most $1/4$.
\end{proof}

The proof of Theorem~\ref{thm:structural} then follows from the last three lemmas. Once the tester $\alg{A} = (G,\tau)$ is fixed, each of the three conditions guarantees correctness (w.p at least $3/4$) on $P = U_n$, $P\in\mathcal{P}_1(\alg{A})$ and $P\in\mathcal{P}_2(\alg{A})$, respectively.

\subsection{Proof of Lemma~\ref{lem:disj_cliques}}
We re-formulate the lemma for the disjoint cliques graph, now followed by a formal proof.
\begin{lemma-repeat}{lem:disj_cliques}
    Fix $n, \eps$. Fix a comparison graph $G = \bigsqcup_{i=1}^{\ell} G_i$, where each $G_i$ is isomorphic to $K_{q}$, for some $q\geq 3$.
    If the following hold for an algorithm $\alg{A} := (G,\tau)$:
     \begin{enumerate}
            \item $q\sqrt{\ell} \geq \frac{\sqrt{12}\sqrt{n}}{\tau \cdot \eps^2}$

            \item $q\sqrt{\ell} \geq \frac{\sqrt{48}\sqrt{n}}{(1-\tau) \cdot \eps^2}$

            \item $q\ell \geq \frac{24\sqrt{n}}{(1-\tau)^2 \eps^2}$
     \end{enumerate}

    then $\alg{A}$ is an $\eps$-uniformity tester that uses $\card{V}$ samples.

\end{lemma-repeat}

\begin{proof}
      We wish to show the three conditions of Theorem~\ref{thm:structural} can be relaxed using the structure of the graph. It is enough to show a lower bound on $\card{E}$, and an upper bound on $c(G)$, which indeed can be done. Whenever $q\geq 3$, we have:
      \begin{align*}
            \card{E} &= \ell\cdot \frac{1^2-1}{2} \geq \frac{\ell\cdot q^2}{3} &
            c(G) &= \ell\cdot \binom{q}{3} \leq \frac{\ell\cdot q^3}{6}
      \end{align*}
      And so
      \[
          \frac{\card{E}^2}{c(G)} \geq \frac{\ell^2 q^4/9}{\ell q^3/6} = \frac{2\ell q}{3} .
      \]
      Which means that we can replace the premise in the main theorem with a the new form:
        \begin{enumerate}
            \item $q\sqrt{\ell} \geq \frac{\sqrt{12}\sqrt{n}}{\tau \cdot \eps^2}$

            \item $q\sqrt{\ell} \geq \frac{\sqrt{48}\sqrt{n}}{(1-\tau) \cdot \eps^2}$

            \item $q\ell \geq \frac{24\sqrt{n}}{(1-\tau)^2 \eps^2}$
        \end{enumerate}
\end{proof}

\section{Limitations of the Method}
\label{sec:lower}
In this section our goal is to better understand the possibilities (and impossibilities) of collision-based testing, in comparison to arbitrary testing methods (many of which were developed for specific testing task, as mentioned earlier in the text).

Here, we use Definition~\ref{def:cb_tester} to shift the discussion to graph terminology, focusing on the comparison graph.
We leverage simple graph properties to show some limitations that apply when generating testers, such as done in Section~\ref{sec:results}.

The results hereinafter apply to testers that can be proven to work using our structural theorem. However, we formulate a simple conjecture (regarding the minimal number of comparisons required for any tester) that implies similar results for any tester that answers Definition~\ref{def:cb_tester}.
In matter of fact, we will formally show that the comparison graphs we have chosen before are essentially the best that could have been chosen for their respective models. If the conjecture is true, this would imply all of these testers are optimal with respect to any tester from Definition~\ref{def:cb_tester}).

We emphasize that in some models, previous works already show tight (or near-tight) lower bounds for arbitrary testers.
Interestingly enough, our optimal collision-based testers achieve optimal (or near-optimal) results even if one is allowed to use an arbitrary method. Thus, the sample complexity we achieve is not only optimal with respect to collision-based testing, but also near-optimal with respect to any uniformity tester. This is evident from a variety of known lower bounds in some of the models considered in this paper (See \cite{Paninski08,DK16,MMO19,DGKR19,ACLST20} for these lower bounds\footnote{More lower bounds in the simultaneous model hold only for the case of a single sample per player, not considered in this paper. e.g.,~\cite{ACHST20,ACT20A}}). The full discussion of collision-based testing vs. arbitrary testers appears at the end of this section, in hope to give a clearer picture of how strong collision-based testing can be.

We first go on to show some basic inequalities that hold in any simple graph (and in our comparison graphs as well). These will serve us for the rest of this section. The goal is to establish the inherent connection between the $3$ sizes: $\card{V}, \card{E}, c(G)$ in any simple graph $G$.

\subsection{Basic Graph Properties}
\label{subsec:graph_prop}
We start by stating the following easy lemma that connects our quantities of interest in any simple graph.
\begin{lemma}\label{lem:graph}
  For any simple graph $G$, it holds that:
  \begin{enumerate}
    \item $\card{E} \leq \frac{\card{V}^2}{2}$

    \item $\card{V} \geq \frac{4\card{E}^2}{2\card{E}+c(G)}$

    \item if $\card{V} \leq \card{E}$ then $c(G) \geq 2\card{E}$

  \end{enumerate}

\end{lemma}

Before proving the lemma, we solve a small mystery that we handled before. In the perfect matching graph, we have a minimal amount of dependencies. In fact, so little, that it overdoes it: it needs to take $\Theta(n/\eps^4)$ samples in order to make the $\Theta(n/\eps^4)$ comparisons needed to apply our structural theorem. This is far from optimal. However, we now understand this pathology, via the following corollary:
\begin{corollary}\label{cor:graph}
  Whenever $\card{V} \leq \card{E}$, we have:
  \[
    \card{V}\cdot c(G) \geq  2\card{E}^2
  \]
\end{corollary}
\begin{proof}[Proof of Corollary \ref{cor:graph}]
    Using items 2 and 3 of lemma \ref{lem:graph}, we see the following: whenever $\card{V} \leq \card{E}$, we know $c(G) \geq 2\card{E}$, and plugged in the second item:
    \[
        \card{V} \geq \frac{4\card{E}^2}{2\card{E}+c(G)} \geq \frac{4\card{E}^2}{2c(G)} = \frac{2\card{E}^2}{c(G)}
    \]
\end{proof}
It is now obvious form the corollary, that we have some tension: On the one hand, we need $c(G)$ to be small enough apply Theorem \ref{thm:structural}. On the other hand, we do not want $c(G)$ to be \emph{too} small, as this would force us to have a large amount of samples -- the value that we actually wish to minimize.

We end this section with the proof of the lemma
\begin{proof}
  The first item is true because the graph is simple (no loops or double-edges).

  For the second item, let us denote by $d_v$ the degree of vertex $v$. It is known that $2\card{E} = \sum_{v\in V} d_v$. It is also clear that the set of 2path subgraphs can be chosen by taking the middle vertex of the 2path, and then choosing the 2 edges of the path (the order matters). So each vertex $v$ is the middle of exactly $d_v(d_v-1)$ such 2paths, and overall we have $c(G) = \sum_{v\in V} d_v^2 - d_v$. We now use Cauchy-Schwarz inequality:
  \[
    (2\card{E})^2 = \left(\sum_{v\in V} d_v\right)^2 \leq\ \left(\sum_{v\in V} 1^2 \right) \left(\sum_{v\in V} d_i^2 \right) = \card{V} \cdot (2\card{E} + c(G))
  \]
  which proves the second item.

  For the third item, we plug the premise in the second item:
  \[
    \card{E} \geq \card{V} \geq \frac{4\card{E}^2}{2\card{E}+c(G)}
  \]
  which means $2\card{E}+c(G) \geq 4\card{E}$, or simply $c(G) \geq 2\card{E}$.
\end{proof}

\subsection{Conditional Impossibility Results}
In order to understand the scope of possible applications of Theorem~\ref{thm:structural}, we stick to graph notations, and combine the $3$ different types of assertions concerning our comparison graph:
\begin{enumerate}
  \item The requirements needed to apply Theorem~\ref{thm:structural}.
  \item The basic graph properties of Lemma~\ref{lem:graph} and Corollary~\ref{cor:graph}.
  \item Individual assertions that apply for the model at hand.
\end{enumerate}

We note that all the results below only use the first requirement of Theorem~\ref{thm:structural}, which asks for the total number of comparisons to be large enough. We believe this requirement to be inherent for \emph{any} collision-based tester to be able to test uniformity. We formulate this belief as the following conjecture, which would strengthen the impossibility results of this section to be independent of the analysis (instead of limitation of Theorem~\ref{thm:structural}, we get a lower bound for all collision-based testers as defined in Definition~\ref{def:cb_tester}).
\begin{conjecture}
\label{conj:edges}
    Any collision-based tester sa defined in Definition~\ref{def:cb_tester} that tests for uniformity with error $1/4$ must have
    \[
        \card{E} = \Omega\left(\frac{n}{\eps^4}\right)
    \]
\end{conjecture}

An argument to support this conjecture is the following: on the one hand, in the case of no dependencies at all (meaning, a new pair of samples is used for each comparison), the conjecture holds. This case is essentially equivalent to another known problem: the problem of finding the bias of a coin (as mentioned in~\cite{ACT20B}). In this problem, a coin is known to have probability of either $1/n$ or at least $(1+\eps^2)/n$, and we need to minimize the number of tosses to differ the two cases. It is known that at least $\Theta\left(n/\eps^4\right)$ tosses are needed.

On the other hand, for the other extreme -- a clique graph, having maximal dependency -- the existing lower bound for uniformity testing (say, the one in~\cite{Paninski08}) shows that $\card{E} = \Theta(\card{V}^2) = \Omega(n/\eps^4)$. Thus, we believe that in every case in between the conjecture should hold as well, even though proving it appears to be a bit tricky.

\subsubsection{Standard Processors}
\label{subsec:lower_standard_processors}
\paragraph{Centralized model.}
As a warm-up, we turn to the classic model, with no restrictions on the graph.
We can easily see that Conjecture~\ref{conj:edges} helps. Using the the first of our basic graph properties, gives $q = \card{V} \geq \sqrt{2\card{E}} = \Omega(\sqrt{n}/\eps2)$, which is the lower bound known to hold for any uniformity tester.\footnote{An interesting fact, just for the classic model, is that the third condition of Theorem~\ref{thm:structural} independently implies the correct lower bound: using corollary \ref{cor:graph}, whenever $\card{V} \leq \card{E}$, we have
\[
    \card{V} \geq \frac{2\card{E}^2}{c(G)} = \Omega\left( \frac{\sqrt{n}}{\eps^2} \right)
\]
}

\paragraph{Simultaneous model.}
In the simultaneous case, the main restriction is that of each player has her own comparison graph, as no comparisons can be made between two different players. This leads us to the following:
\begin{corollary}\label{cor:lb_simultaneous}
    Assuming Conjecture~\ref{conj:edges} holds, the number of samples per player of any collision-based uniformity tester in the simultaneous model is
    \[
        q' = \Omega\left(\frac{\sqrt{n}}{\sqrt{k} \cdot \eps^2}\right) .
    \]
\end{corollary}

\begin{proof}
      We start by noting that in this model, our sample complexity is the maximum amount of samples one player draws. Formally, we write the comparison graph as the union of $k$ disjoint parts: $G = \bigsqcup_{i\in[k]} G_i$, where $G_i = (V_i, E_i)$, and so the complexity measure is simply $q' = max_{i\in[k]} \card{V_i}$.

      We note, however, that $\card{E_i} \leq \card{V_i}^2 / 2$ for each component $G_i$, and therefore:
      \[
        \card{E} = \sum_{i\in[k]} \card{E_i} \leq \sum_{i\in[k]} \card{V_i}^2 / 2 \leq k\cdot max_{i\in[k]} \card{V_i}^2/2 = k\cdot q'^2/2 .
      \]
      Plugging in Conjecture~\ref{conj:edges}, we get
      \[
        q' \geq \sqrt{\frac{2\card{E}}{k}} = \Omega\left(\frac{\sqrt{n}}{\sqrt{k} \cdot \eps^2}\right) ,
      \]
      which ends the proof.
\end{proof}

\paragraph{Asymmetric cost model.}
The more elaborate version of the asymmetric-cost model also impose similar limitations, where the difference comes from the generalized definition of the complexity measure.
\begin{corollary}\label{cor:lb_simultaneous_asymmetric}
      We observe the asymmetric-cost simultaneous model, with sampling rate vector $(R_1,\dots,R_k)$. If Conjecture\ref{conj:edges} holds, then any collision-based uniformity tester in this model must use sampling time of
      \[
        t = \Omega\left(\frac{\sqrt{n}}{\eps^2\lnorm{2}{R}}\right) .
      \]
\end{corollary}

\begin{proof}
      We again use $G = \bigsqcup_{i\in[k]} G_i$, where $G_i = (V_i, E_i)$.
      However, the complexity measure is the time $t$ in which player $i$ with rate $R_i$ can obtain $\card{V_i} = q_i = t\cdot R_i$ samples.

      Again, using the trivial edges-vertices inequality over each component $G_i$, we get
      \[
        \forall i.\ \card{E_i} \leq \card{V_i}^2 / 2 = t^2\cdot R_i^2 / 2
      \]
      and summing all together, we get
      \[
        \card{E} = \sum_{i\in[k]} \card{E_i} \leq \sum_{i\in[k]} t^2\cdot R_i^2 / 2 \leq t^2 \lnorm{2}{R}^2 / 2 .
      \]
      Joined with Conjecture~\ref{conj:edges}, it concludes the proof
      \[
        t \geq \sqrt{\frac{2\card{E}}{\lnorm{2}{R}^2}} = \Omega\left(\frac{\sqrt{n}}{\eps^2\lnorm{2}{R}}\right) .
      \]
\end{proof}

\subsubsection{Memory-Constrained Processors}
\label{subsec:lower_streaming_processors}
Here we use a slightly more sophisticated argument, to show that the memory constraint can also be translated to graph notation. We emphasize that our desire is to show limitations of our framework, and so we relax the model and assume that comparisons are made on designated memory cells, in which we can only store $m'$ element names.
In order to count collisions \emph{accurately}, we cannot expect to compress this data further. For example, $m' = o(\sqrt{n})$ and samples are drawn from the uniform distribution $P = U_n$, with $99\%$ we need to write in our memory $m'$ different elements, and this information cannot be compressed.

We go on to show how the memory constraint translates well:
\begin{claim}\label{claim:streaming_avg_deg}
      Let us assume a constrained machine can only store $m'$ elements at a time, and it is able to accurately count collisions on a comparison graph $G = (V,E)$. Then it must be the case that
      \[
        \card{E} \leq m' \cdot \card{V}
      \]
\end{claim}

\begin{proof}
    w.l.o.g let us name the vertices, or samples, by their order in the stream $V = \set{v_1,v_2,\dots,v_s}$, and w.l.o.g let us think of the edges in $E$ as ordered pairs (We only write $(i.j)\in E$ for pairs where $i < j$).

    Now, we note that at time $t$, upon processing the sample $v_t$, the memory can only store $m'$ samples from the set $s(v_1),\dots,S(v_{t-1})$. This means that the number of edges in $E$ of the form $(v_i,v_t)$ is at most $m'$. Now, we can count our (ordered) edges using the second item:
    \[
        \card{E} = \sum_{v_j\in V} \card{\set{(u,v_j) \in E \st u\in V}} \leq \sum_{v_t\in V} m' = \card{V}\cdot m'
    \]
    which completes the proof.
\end{proof}

\paragraph{Centralized model with memory constraints.}
We next apply this claim to give similar lower bounds in the following models.
\begin{corollary}\label{cor:lb_streaming}
  Assume Conjecture~\ref{conj:edges} holds. In order to test uniformity in a memory-constrained machine, using a collision-based tester and storage of up to $m' = \Theta\left(m/\log(n)\right)$ samples at any given time, one must have
  \[
    q \geq \Omega\left(\frac{n\log(n)}{m\eps^4}\right) ,
  \]
  where $q$ is the total amount of samples used by the algorithm.
\end{corollary}
\begin{proof}
    As our measure complexity is the total amount of samples, $q = \card{V}$, we can combine the claim above with Conjecture~\ref{conj:edges} to get:
    \[
        q = \card{V} \geq \frac{\card{E}}{m'} = \Omega\left(\frac{n\log(n)}{m\eps^4}\right) ,
    \]
    which ends the proof.
\end{proof}

\paragraph{Simultaneous model with memory constraints.}
We observe the combined model of simultaneous model with memory-constrained machines. We again restrict ourselves to a specific framework: each player uses her own graph $G_i = (V_i,E_i)$, where memory is allocated for sampled elements, and then all players send a short message to the referee. The entire comparison graph the algorithm is based on is $G = \bigsqcup_i G_i$.
\begin{corollary}\label{cor:lb_simultaneous_streams}
  Any collision-based uniformity tester in a simultaneous model of $k$ machines that can store up to $m'$ samples each, must use
  \[
    q' = \Omega\left( \frac{n}{k \eps^4} \cdot \max \set{\frac{\log(n)}{m},\frac{\sqrt{k}\eps^2}{\sqrt{n}}} \right)
  \]
  samples per player, assuming Conjecture~\ref{conj:edges} holds .
\end{corollary}

\begin{proof}
  We start be re-writing the desired expression:
  \[
     q' = \Omega\left( \max \set{\frac{\sqrt{n}}{\sqrt{k} \eps^2}, \frac{n\log(n)}{m k \eps^4}} \right)
  \]

  We show the two lower bounds separately, resulting in a lower bound of the maximum term.

  Indeed, the first lower bound can be derived directly from Corollary~\ref{cor:lb_simultaneous}:
  \[
    q' = \Omega\left(\sqrt{n/(k\eps^2)}\right)
  \]

  For the second lower bound, we extend Corollary~\ref{cor:lb_streaming} instead. As each $G_i$ is done by a machine with memory constraints, we apply Claim~\ref{claim:streaming_avg_deg} to player $i$ and get $\card{V_i} \geq \card{E_i}/m'$. We recall that our measure complexity is in fact $q' := \max_i \card{V_i}$, and as the maximum is greater than the average, we get:
  \[
    q' \geq \frac{\sum_i \card{V_i}}{k} \geq \frac{\sum_i \card{E_i}}{m'\cdot k} = \frac{\card{E}}{m'\cdot k}
  \]
  And plugging in Conjecture~\ref{conj:edges} on the entire graph $G=(V,E)$, we get
  \[
    q' = \Omega\left( \frac{n}{m' \cdot k \eps^4} \right) = \Omega\left( \frac{n\log(n)}{m \cdot k \eps^4} \right),
  \]
  concluding the proof.
\end{proof}

\subsection{Discussion}
\label{subsec:discussion}
We point out to the fact that in all $4$ models (as well as the classic model), the limitation of this method coincide with the upper bounds we obtained in the previous section. This does not come as a surprise, since the choice of the ``right'' comparison graph is easily made once the constraints each model imposes are understood. The guiding rule is rather straightforward: compare all pairs that can be compared.

It is more interesting, though, to see how well the collision-based testers perform compared to known impossibility results (ones which apply for any uniformity tester, and not only collision-based). It turns out for the most part, these testers compete well with others.

For the classical model, as already established in~\cite{DGPP17}, collision-based testing is in fact optimal in all parameters.

For the two simultaneous models (with no memory constraints), the collision-based testers perform optimally in terms of $n,k$ (or $n, (R_1,\dots,R_k)$ in the assymetric case). To the best of our knowledge, the only testers for these models that consider multiple samples per processor are the ones of~\cite{FMO18}. The same asymptotic sample complexity is obtained in both papers (for both models), but there are two non-trivial differences. On the downside, the new testers use $\log\left(\Theta(1/\eps^4)\right)$ bits per communication, instead of a single one used in~\cite{FMO18}. On the upside, the new testers work for the full range of the parameter $k$ (the number of players), whereas the previous results excluded extreme values: e.g., in the symmetric case it only works for $c/\eps^4 \leq k \leq c'n\eps^4$.

Lower bounds for these two models are shown in~\cite{MMO19}, both for a single bit, and the general case of $\ell$-bit messages. While the tester of~\cite{FMO18} is an optimal one-bit protocol, ours is not known to be optimal $\ell$-bit protocol.
This is true as the aforementioned lower bound weakens by a factor $2^{\ell}$ for the longer $\ell$-bit messages. In our case, we have $2^{\ell} = \Theta(1/\eps^4)$, which means there is a gap of $\eps^4$ between the general lower bound, and the optimal collision-based tester we obtain.\footnote{Despite the gap, no better tester is known for this amount of bits and $q > 1$ samples per player. It still might be the case that collision-based testers are optimal in this regime}

In the streaming model our tester attain the same sample complexity as the best known tester (that of~\cite{DGKR19}), but for a wider range for the parameter $m$. A matching lower bound can also be found in~\cite{DGKR19}, but only for a more restricted range of the parameter $m$, whereas for the general case they give a weaker lower bound, which leaves an $\eps^2$ gap between the general case and the optimal collision-based tester.

To the best of our knowledge, the simultaneous model with memory constraints was never considered in the context of distribution testing, and therefore there are no prior results. However, we conjure that similarly to before, the collision-based tester achieves optimal sample complexity in parameters $n, k, m$, with a possible $\poly(\eps)$ gap that would pop, as before, due to the use of longer messages.


\appendix

\section{Is every edge a blessing?}
\label{sec:counter_example}
In our results the guiding rule in any model was ``compare every pair the model allows you to''. However, this intuitive rule was not formalized, and it could be for a good reason. Earlier, in Section~\ref{subsec:discussion}, we posed the question of whether adding another edge (another comparison) to the entire average, can hurt our assessment of the collision probability, and thus weaken our tester.

We leave it as an open question whether the actual reliability of the tester can only get better when adding edges (or rather, under which circumstances it is true). For now, we do point out a pathological yet surprising example where a comparison graph $H$ that fills the requirements of Theorem \ref{thm:structural} for some constant threshold $\tau$, is a subgraph of $G$ that is not good for any constant threshold $\tau'$.

Obviously, the number of edges only goes up when adding edges, so for such an example we should focus on breaking the third condition in Theorem \ref{thm:structural}, by making the supergraph $G$ one with a large amount of $2$-path. Indeed, this is possible:
Consider $H$ to be the simple cycle of size $b\cdot\frac{n}{\eps^4}$, for some constant $b$. It holds that $c(H) = 2\card{E_H} = 2\card{V_H} = 2b\cdot\frac{n}{\eps^4}$. Taking large enough $b$ (as a function of $\tau$) would give us $H$ that fills the three conditions of the theorem, as $c(H) / \card{E_H}^2 = 1 / (2\card{E_H}) = \frac{\eps^4}{2n}$. So indeed $(H,\tau)$ is a good tester.

The supergraph $G$ of $H$ would be this: we choose one the vertex $v$ in $H$, and connect it with all other edges. Now $G$'s edges are either part of a star or a cycle. $\card{E_G} = 2\card{V_G}-3 \leq 2\card{V_G}$, but the star gives us $c(G) \geq \card{V_G}^2/3$ (Solely by the 2-paths with $v$ as the middle vertex). This means that
\[
    \frac{c(G)}{\card{E_G}^2} \geq  \frac{\card{V_G}^2/3}{(2\card{V_G})^2} = \frac{1}{12}
\]
So for \emph{any} constant threshold $\tau'$ and asymptotic values of $n, \eps$, Theorem \ref{thm:structural} cannot be applied for $(G,\tau')$.

\end{document}